\documentclass[11pt]{article}

\usepackage[letterpaper,margin=1in]{geometry}
\usepackage{amsmath,amssymb,amsfonts,amsthm}
\usepackage{graphicx}
\usepackage[round,authoryear]{natbib}
\usepackage[hidelinks]{hyperref}
\usepackage{url}

\usepackage{times}
\usepackage{tikz}
\usetikzlibrary{
  positioning,
  arrows.meta,
  decorations.pathreplacing
}
\usepackage{subfig}
\usepackage{bm}
\usepackage{mathtools}
\usepackage{algorithm}
\usepackage{algpseudocode}
\algrenewcommand\textproc{\texttt}
\usepackage{paralist}
\usepackage{xcolor}
\usepackage{afterpage}
\usepackage[utf8]{inputenc}
\usepackage{enumitem}
\usepackage{multirow}
\usepackage{booktabs}
\usepackage{siunitx}

\newcommand{\field}[1]{\mathbb{#1}}
\newcommand{\R}{\field{R}}

\newcommand{\N}{\field{N}}

\newcommand{\E}{\field{E}}

\newcommand{\order}{\mathcal{O}}

\newcommand{\prob}{\mathbb{P}}

\newtheorem{result}{Result}
\newtheorem{assumption}{Assumption}

\newtheorem{theorem}{Theorem}
\newtheorem{lemma}{Lemma}
\newtheorem{remark}{Remark}
\newtheorem{corollary}{Corollary}

\def\var{{\mbox{var}}}

\begin{document}

\title{Predictive Subsampling for Scalable Inference in Large Networks}

\author{
Arpan Kumar
\and
Minh Tang
\and
Srijan Sengupta
}

\date{}

\maketitle

\begin{center}
    Department of Statistics\\
       North Carolina State University
\end{center}

\begin{abstract}
Current methods for statistical inference in networks often encounter substantial computational bottlenecks when applied to the massive network datasets that are increasingly common across scientific domains.
In this paper, we develop \textit{Predictive Subsampling} (\texttt{PredSub}), a scalable framework for estimation and two-sample testing in networks. The central idea is to replace a full-sample estimation procedure by estimation on a random subsample, followed by out-of-sample prediction of the remaining vertices. This construction exploits the fact that the subsample provides an inferential anchor, which enables each remaining vertex to be incorporated in a \textit{predictive} manner through a fast vector operation. Building on this estimator, we develop two scalable procedures for two-sample testing, namely \texttt{PredSubTest} and \texttt{PureSubTest}. We establish finite-sample error bounds as well as estimation and testing consistency of the proposed methods in both Frobenius and two-to-infinity norms.  These results formally characterize the trade-offs between statistical accuracy and computational efficiency with respect to the subsample size, the choice of test statistic, and the choice of norm.
We demonstrate the empirical performance of the proposed methods through detailed simulation studies and two real-world applications involving DBLP coauthorship networks and the Cannes 2013 social media networks.
\end{abstract}

\small\textbf{Keywords:}
Estimation; 
Hypothesis testing;
Random dot product graph model; 
Generalized random dot product graph model;
Adjacency spectral embedding.

\section{Introduction}

Over the last two decades, a substantial body of statistical methodology has been developed for  inference with network data \citep{goldenberg2010survey, kolaczyk2014statistical, newman2018networks, athreya2018statistical, sengupta2025statistical}. In this paper, we focus on two fundamental inferential tasks that have been central to this literature. The first is estimation, where the goal is to estimate population parameters from an observed network \citep{bickel2009nonparametric,sussman2012consistent,sengupta2018block}. The second is two-sample testing, where the goal is to determine whether two observed networks were generated from the same model \citep{tang2017semiparametric,levin2017central,ghoshdastidar2020two,bhadra2025bootstrap}.

Despite this progress, the lack of computational scalability remains a bottleneck to statistically principled network inference in many application domains. Many existing methods for estimation and two-sample testing are statistically well motivated, but they rely on matrix operations that become prohibitively expensive for large-scale networks \citep{tang2017semiparametric,zhang2022distributed,chakraborty2025scalable}.
 This difficulty is especially pronounced for bootstrap-based tests where such expensive matrix operations must be repeated many times. 
 For example, as shown in Tables~\ref{tab:real_coauth} and~\ref{tab:real_cannes}, applying the two-sample testing procedure of \citet{bhadra2025bootstrap} to the Coauth-DBLP dataset of \citet{Benson-2018-simplicial} and the Cannes 2013 dataset of \citet{omodei2015characterizing} requires 22.7 hours and 28.1 hours, respectively. Such high runtime substantially limits the use of existing methods in scientific domains where large networks are increasingly common, including connectomics, biomedical text mining, and social media \citep{roncal2013migraine,komolafe2021scalable}.

To address this scalability bottleneck, we propose a general strategy for scalable network inference called \textit{Predictive Subsampling} (\texttt{\texttt{PredSub}}).
The proposed method
consists of three steps:
(1) select a subsample of vertices from the full network;
 (2) carry out the relevant matrix operation on the subgraph formed by the subsampled vertices;
 (3) add the remaining vertices via out-of-sample prediction.
 Figure \ref{fig:PredSub} provides a schematic.
Computationally, the key idea is to replace the large-scale matrix operation required by existing approaches with a much smaller matrix operation (Step 2) followed by a large number of vector operations which can be implemented in parallel (Step 3).
Since vector operations are very fast compared to matrix operations, \texttt{PredSub} drastically reduces computational cost.
Please see Remark \ref{remark:compcost} for a more formal discussion on computational efficiency.

\begin{figure}[ht]
    \centering
    \resizebox{0.6\textwidth}{!}{
    \begin{tikzpicture}[
    thick,
    fullbox/.style={rectangle, draw=black, minimum width=10cm, minimum height=0.8cm, fill=green!40},
    halfbox/.style={rectangle, draw=black, minimum width=4.5cm, minimum height=0.8cm, fill=blue!30},
    emptybox/.style={rectangle, draw=black, minimum width=5.2cm, minimum height=0.8cm, fill=yellow!50},
    graybox/.style={rectangle, draw=black, minimum width=4.5cm, minimum height=0.8cm, fill=gray!50},
    smallbox/.style={rectangle, draw=black, minimum width=0.4cm, minimum height=0.8cm, fill=gray!50},
    arrow/.style={-Stealth, thick}
]

\coordinate (origin) at (0,0);

\node[fullbox, anchor=west] (graph) at (origin) {Graph};

\node[halfbox, below=1.8cm of graph.west, anchor=west] (subgraph) {Subgraph};
\node[emptybox, right=0.3cm of subgraph] (rightspace) {Out of Sample};

\draw[arrow] ([xshift=-79pt]graph.south) -- ([yshift=2pt]subgraph.north) node[midway, right=5pt] {Sampling};

\node[left=0.5cm of subgraph] {\textbf{Step 1}};

\node[graybox, below=1.8cm of subgraph.west, anchor=west] (model) {Estimated Submodel};

\draw[arrow] (subgraph.south) -- (model.north) 
node[midway, right=1pt] {Matrix operation};

\node[left=0.5cm of model] {\textbf{Step 2}};

\node[graybox, below=1.8cm of model.west, anchor=west] (predleft) {};

\foreach \i in {0,...,8} {
    \node[smallbox, right=0.3cm of predleft.east, xshift=\i*0.6cm, anchor=west] (p\i) {};
}

\draw[arrow] ([yshift=-2pt]rightspace.south) -- ([yshift=2pt]p4.north) node[midway, left=3pt] {Prediction}
    node[midway, right=3pt] {(vector operation)};

\node[left=0.5cm of predleft] {\textbf{Step 3}};

\draw[decorate,decoration={brace,mirror,amplitude=6pt}, thick]
    ([yshift=-0.2cm]p0.south west) -- ([yshift=-0.2cm]p8.south east) 
    node[midway, below=6pt] {Add remaining nodes};

\end{tikzpicture}
}
    \caption{\texttt{PredSub} Schematic}
    \label{fig:PredSub}
\end{figure}

From a statistical perspective, the methodological intuition behind \texttt{PredSub} is as follows. Consider an out-of-sample vertex, and let $y$ denote the vector of observed edges between this vertex and the subsampled vertices. Under a general low-rank network model, the distribution of $y$ is governed by two components: a \textit{local} parameter specific to the vertex of interest, and a \textit{global} parameter that acts as its structural \textit{anchor} to the subgraph; see Equation \eqref{eq:predsub} and Remark \ref{remark:localglobal} for more details. Estimating both components from $y$ alone is generally infeasible. \texttt{PredSub} resolves this difficulty by first estimating (in Step 2) the global parameter from the sampled subgraph through a small matrix computation. Once this anchor parameter has been accurately estimated, the vertex-specific parameter can be estimated from $y$ in Step 3 using only a vector computation.

This paper makes the following contributions. 
We develop \texttt{PredSub} as a scalable estimation procedure, which inherits the statistical accuracy of the full-sample estimator under a general low-rank modeling framework while substantially reducing computational cost.
Next, we propose two scalable procedures for two-sample testing.
The first procedure is called \texttt{PredSubTest}, which uses \texttt{PredSub} to compute the test statistic within a parametric bootstrap framework.
A second, even faster procedure is called \texttt{PureSubTest}, which applies bootstrap-based testing only to the subsample. We establish finite-sample error bounds and consistency results for both estimation and testing, and characterize classes of alternatives under which the proposed tests are the most effective. These theoretical results explicitly formalize the trade-offs between statistical accuracy and computational cost.

In related work, recent years have witnessed important advances in  scalable network inference \citep{bhadra2026scalable}.
Most of this literature has focused on community detection, where pseudo-likelihood methods \citep{wang2021fast}, distributed algorithms \citep{zhang2022distributed,wu2023distributed}, and divide-and-conquer strategies \citep{mukherjee2021two, chakrabarty2023sonnet} have been developed as computationally efficient alternatives to clustering the full network.

For scalable estimation and two-sample testing, to the best of our knowledge, the only existing work is by \citet{chakraborty2025scalable} who proposed a subsampling-based divide-and-conquer approach.
While their method achieves substantial computational savings, it still requires repeated matrix operations, as estimation is performed separately on multiple subgraphs and the resulting local estimates are combined via \textit{stitching}.
By contrast, the scalable alternatives proposed in this paper require only one matrix operation, leading to greater computational efficiency. 
Please see the numerical results in Sections~\ref{sec5} and~\ref{sec6} for a detailed performance comparison.

The remainder of the paper is organized as follows. In Section~\ref{sec3}, we introduce the proposed estimation and testing procedures. In Section~\ref{sec4}, we present the theoretical results on error bounds and consistency. Simulation studies under various settings are reported in Section~\ref{sec5}, followed by applications to the Coauth-DBLP and Cannes 2013 datasets in Section~\ref{sec6}. Section~\ref{sec7} concludes the paper with a discussion. All technical proofs are in Appendix~\ref{appA}.
All codes can be found at the following {GitHub repository}: \url{https://github.com/ArpanK72/PredSub/tree/main}.

\section{Methodology}
\label{sec3}

In this section, we introduce the notation used throughout the paper, describe the statistical framework, and present the proposed estimation and testing procedures.

\subsection{Notations} \label{sec:notation}

We use the symbol $:=$ to assign mathematical definitions. The set of all positive integers is denoted by $\mathbb{N}_+$. For two non-negative sequences $\{a_n\}_{n \in \mathbb{N}_+}$ and $\{b_n\}_{n \in \mathbb{N}_+}$, we write $a_n  = \Theta(b_n)$ if both are of the same order, i.e., there exists absolute constant $c_1,c_2 > 0$ such that $c_1b_n\leq a_n \leq c_2b_n$ for all $n \in \mathbb{N}_+$. We use $c_0,c_1,c_2,\cdots, C_1,C_2,\cdots, K_1,K_2,\cdots$, etc., to denote constants that are independent of $n$. If $a_n / b_n$ remains bounded as $n\to\infty$, we write $a_n = \order(b_n)$ or $b_n = \Omega(a_n)$, and if $a_n / b_n \to 0$, we write $a_n = o(b_n)$ or $b_n = \omega(a_n)$. 
For any $d \in \mathbb{N}_+$, we denote by $I_d$ the $d \times d$ identity matrix. Similarly, $I_{p,q}$ denote a diagonal matrix with $p$ many $1$'s followed by $q$ many $-1$'s in its diagonal. $\mathbb{O}(d)$ is the set of all $d \times d$ orthogonal matrices and $\mathbb{O}(p,q) = \{M\in\mathbb{R}^{(p+q)\times(p+q)}: M^\top I_{p,q}M = I_{p,q}\}$ is known as the indefinite orthogonal group.

Let $M$ be any $n\times n$ symmetric matrix with rank $d$. Then, it has $d$ positive singular values $\sigma_1(M)\geq \sigma_2(M)\geq \cdots \geq \sigma_d(M)$ and the condition number $\kappa(M)$ of $M$ is defined as,
$\kappa(M) := \sigma_1(M)/\sigma_d(M)$
i.e., the ratio of the largest to the smallest (non-zero) singular values of $M$. Note that if $M$ is positive semi-definite with rank $d$ then the condition number becomes $\kappa(M) := \lambda_1(M)/\lambda_d(M)$,
where $\lambda_1(M) \geq \lambda_2(M) \geq \cdots \geq \lambda_d(M) $ are the $d$ positive eigenvalues of $M$.

We will also use the following matrix norms in our paper. The spectral norm of $M$ is defined as $\|M\| :=\sigma_1(M),$ and the Frobenius norm is given by $\|M\|_F := \bigl(\operatorname{trace}(M^\top M)\bigr)^{1/2}.$ The two-to-infinity norm of a matrix $M$ is defined as $\|M\|_{2 \to \infty} := \max_{j \in [n]} \|M_{j.}\|$ where $M_{j.}$ denote the $j$th row of matrix $M$.

\subsection{Statistical framework}
Let $A$ be the $n \times n$ symmetric adjacency matrix of a simple, undirected graph $G$ with vertex set $V = \{1,2,\dots,n\}$.  We assume that $A \thicksim \text{Bernoulli}(P)$, i.e., the $A_{ij}$’s for $i < j$ are independent and there is an underlying edge-probability matrix $P\in[0,1]^{n\times n}$ such that $A_{ij} \thicksim \text{Bernoulli}(P_{ij})$.
We further assume that $P$ is low rank, with $p$ positive eigenvalues and $q$ negative eigenvalues, where $d=p+q$ remains fixed as $n$ grows. Then $P$ admits the factorization
\begin{equation}
    P = XI_{p,q}X^\top,
    \label{eq:grdpg}
\end{equation}
where $X\in\mathbb{R}^{n\times d}$ is the latent position matrix.
Note that this is a very flexible modeling framework since we do not impose any structural assumptions on $P$ other than it being low-rank.
Model \eqref{eq:grdpg} is called the generalized random dot product graph (GRDPG) model and it includes many commonly used network models as special cases, including the stochastic blockmodel and its generalizations \citep{rubin2022statistical, koo2023popularity,Bhadra2026unified}.
A notable special case is the random dot product graph (RDPG) model, where $q=0$.

In the next two subsections, we describe the proposed methods for scalable estimation and two-sample testing under this modeling framework.

\subsection{Scalable estimation with \texttt{PredSub}}\label{sec:estimation}

We first recall the adjacency spectral embedding (ASE) estimator under the GRDPG framework. Let $D_A$ be the diagonal matrix containing the $d$ eigenvalues of $A$ with largest magnitudes, and let $U_A$ be the $n\times d$ matrix whose columns are the corresponding eigenvectors. We order the $d$ selected eigenvalues so that the positive eigenvalues appear first and the negative eigenvalues appear second, and define
    $D_A = \mathrm{diag}(D_{A+},D_{A-})$, 
    $U_A = [\,U_{A+}\mid U_{A-}\,]$,
where $D_{A+}$ and $D_{A-}$ contain the selected positive and negative eigenvalues, respectively. Let $\hat p$ and $\hat q$ denote the number of selected positive and negative eigenvalues, so that $\hat p+\hat q=d$. The ASE estimator is then defined as
    $\widehat X = U_A |D_A|^{1/2}$,
where $|D_A|$ is obtained by replacing the diagonal entries of $D_A$ by their absolute values. The corresponding estimate is
    \begin{equation}
       \widehat P = \widehat X I_{\hat p,\hat q}\widehat X^\top. 
       \label{eq:ase}
    \end{equation}    
Throughout this paper, we will refer to the above procedure as the ASE procedure and the estimator in Equation \eqref{eq:ase} as the ASE estimator. 

The computational complexity of the spectral decomposition required for ASE is $\order(n^3)$, which becomes prohibitive for large-scale networks. In practice if we use randomized SVD or iterative Lanczos-based algorithms, this can be reduced to $\order(n^2d)$ which is still too high for larger networks. 
To address this challenge, we propose the following computationally scalable alternative called \texttt{PredSub}, which consists of three steps.

\begin{enumerate}
    \item \textbf{Subsampling:} In this step we choose a subsample $S \in V$. There are two key aspects of this step: the subsample size, $m$, and the sampling technique, $\mathcal{S}$.
    In this paper, we consider the case where $m << n$, specifically, $m = \Omega((\log n)^{1+a})$ for some $a>0$; see Section \ref{sec:estimation-theory} for more details.
    Our theoretical analysis is based on uniform random sampling as the sampling technique $\mathcal{S}$. 
    We note that other sampling techniques could be used in this step as well.
    
    Let $S^c = V \backslash S$ denote the set of out-of-sample vertices,
    and write 
    $$X = [X_S^\top\mid X_{S^c}^\top]^\top,$$ 
    where $X_S$ and $X_{S^c}$ have $m$ and $n-m$ rows, respectively. Then we can express $P$ and $A$ in the following partitioned form:
$$P = \begin{pmatrix}
    X_{S}I_{p,q}X_{S}^\top & X_{S}I_{p,q}X_{S^c}^\top\\
    X_{S^c}I_{p,q}X_{S}^\top & X_{S^c}I_{p,q}X_{S^c}^\top
\end{pmatrix} = \begin{pmatrix}
    P_S & P_{S,S^c} \\
    P_{S^c,S} & P_{S^c}
\end{pmatrix}, \qquad 
A = \begin{pmatrix}
    A_S &  A_{S,S^c} \\
    A_{S^c,S} & A_{S^c}
\end{pmatrix}.$$

    \item \textbf{Subgraph Inference:} In this step, we apply the ASE procedure on the subgraph $A_S$, spanned by the vertices in $S$, to obtain the estimate $\hat{X}_S$. This step depends on only one input parameter, the rank ($d$) of the underlying probability matrix, which can be estimated if unknown.

    \item \textbf{Predictive Estimation:} In this step, we estimate $X_{S^c}$ in a \textit{predictive} manner using $\widehat X_S$ (from Step 2) and $A_{S^{c},S}$, the observed adjacency submatrix of cross-edges between $S^c$ and $S$. Note that this step is related to out-of-sample embedding which was previously studied by \cite{levin2021limit} under the RDPG model.

    Consider any vertex $i\notin S$. Let $A_{i,S}$ denote the $1\times m$ vector of observed edges between vertex $i$ and the sampled vertices, and let $P_{i,S}=\E[A_{i,S}]$. Under the GRDPG model, we can write
    \begin{equation}
        P_{i,S} = x_i^\top I_{p,q}X_S^\top .
        \label{eq:cross_edge_identity}
    \end{equation}
    Provided that $X_S^\top X_S$ is nonsingular, we have
        $$P_{i,S}X_S = x_i^\top I_{p,q}X_S^\top X_S
        \Rightarrow
        x_i^\top = P_{i,S}X_S (X_S^\top X_S)^{-1}I_{p,q}.$$      
Replacing the unknown quantities on the right-hand side of the final expression by their empirical counterparts leads us to the estimate
\begin{equation}
     \widehat x_i^\top
        =
        A_{i,S}\widehat X_S
        (\widehat X_S^\top \widehat X_S)^{-1}I_{p,q},
         \label{eq:predsub}
\end{equation}
 and this computation can be carried out  in parallel for each $i\notin S$.
 Note that the matrix $\widehat X_S
        (\widehat X_S^\top \widehat X_S)^{-1}I_{p,q}$ is independent of $i$, and therefore only needs to be computed once.
    Stacking $\widehat x_i^\top$ over all $i\notin S$ yields    
     $ \widehat X_{S^c}
        =
        A_{S^c,S}\widehat X_S
        (\widehat X_S^\top\widehat X_S)^{-1}I_{p,q}$.   
\end{enumerate}

Please see Algorithm \ref{algo1} for an algorithmic description of the procedure, and Figure \ref{fig:placeholder} for a visual depiction.

\begin{algorithm}[ht]
\caption{Estimation with Predictive Subsampling}
\label{algo1}
\begin{algorithmic}[1]

\Require Adjacency matrix $A_{n \times n}$, subsample size $m$, rank $d$, sampling technique $\mathcal{S}$
\Ensure Estimated probability matrix $\hat{P}_{PS}$

\Procedure{\texttt{PredSub}}{$A,m,d,\mathcal{S}$}

\State Select a subsample $S$ of $m$ vertices using subsampling technique $\mathcal{S}$

\State Implement ASE on the adjacency submatrix $A_S$ to obtain $\hat{X}_S$

\State Set $\hat{p}$ as the number of positive eigenvalues among the $d$ largest (in modulus) eigenvalues of $A_S$

\State Compute the scaling matrix $
C = \hat{X}_S(\hat{X}_S^\top \hat{X}_S)^{-1} I_{\hat{p}, d-\hat{p}}$

\State Compute $\hat{X}_{S^c} = A_{S^c,S} C$

\State Set $\hat{X}_{PS} = \left[ \hat{X}_S^\top \mid \hat{X}_{S^c}^\top \right]^\top$

\State Output $\hat{P}_{PS} = \hat{X}_{PS} I_{\hat{p}, d-\hat{p}} \hat{X}_{PS}^\top$

\EndProcedure
\end{algorithmic}
\end{algorithm}

\begin{figure}[ht]
    \centering
\resizebox{\textwidth}{!}{
\begin{tikzpicture}[baseline={(current bounding box.center)}]
\node (m) [inner sep=0pt] {
 \(
 \left[
  \begin{array}{ccccccc}
    A_{11} & A_{12} & \cdots & A_{1m} & A_{1,m+1} & \cdots & A_{1,n} \\
    A_{21} & A_{22} & \cdots & A_{2m} & A_{2,m+1} & \cdots & A_{2,n}\\
    \vdots & \vdots & \ddots & \vdots & \vdots & \ddots  & \vdots\\
    A_{m1} & A_{m2} & \cdots & A_{mm} & A_{m,m+1} & \cdots & A_{m,n} \\
    A_{m+1,1} & A_{m+1,2} & \cdots & A_{m+1,m} & A_{m+1,m+1} & \cdots & A_{m+1,n} \\
    A_{m+2,1} & A_{m+2,2} & \cdots & A_{m+2,m} & A_{m+2,m+1} & \cdots & A_{m+2,n} \\
    \vdots & \vdots & \ddots & \vdots & \vdots & \ddots & \vdots\\
    A_{n1} & A_{n2} & \cdots & A_{nm} & A_{n,m+1} & \cdots & A_{n,n} \\
  \end{array}
  \right]
  \)
};

\draw[red, thick, rounded corners]
  ([xshift=10pt, yshift=1pt] m.north west)
  rectangle
  ([xshift=160pt, yshift=61pt] m.south west);

\draw[blue, thick, rounded corners]
  ([xshift=170pt, yshift=1pt] m.north west)
  rectangle
  ([xshift=290pt, yshift=61pt] m.south west);

\draw[blue, thick, rounded corners]
  ([xshift=10pt, yshift=-62pt] m.north west)
  rectangle
  ([xshift=160pt, yshift=45pt] m.south west);

\draw[blue, thick, rounded corners]
  ([xshift=10pt, yshift=-75pt] m.north west)
  rectangle
  ([xshift=160pt, yshift=32pt] m.south west);

\draw[blue, thick, rounded corners]
  ([xshift=10pt, yshift=-105pt] m.north west)
  rectangle
  ([xshift=160pt, yshift=1pt] m.south west);

\draw[black, thick]
  ([xshift=72pt, yshift=1pt] m.north west)
  -- ++(0,10pt)
  -- ++(245pt,0);

\node[red, anchor=west, align=left, text width=7cm, font=\footnotesize]
  at ([xshift=315pt, yshift=5pt] m.north west)
  {Subsample used for ASE
   to estimate corresponding latent positions};

\draw[black, thick]
  ([xshift=160pt, yshift=-70pt] m.north west)
  -- ++(5pt,-7pt)
  -- ++(144pt,0);

\node[blue, anchor=west, font=\footnotesize]
  at ([xshift=315pt, yshift=-78pt] m.north west)
  {Estimating latent position of node $(m\!+\!1)$};

\draw[black, thick]
  ([xshift=160pt, yshift=-84pt] m.north west)
  -- ++(5pt,-7pt)
  -- ++(144pt,0);

\node[blue, anchor=west, font=\footnotesize]
  at ([xshift=315pt, yshift=-92pt] m.north west)
  {Estimating latent position of node $(m\!+\!2)$};

\node[blue, anchor=west, font=\footnotesize]
  at ([xshift=315pt, yshift=-115pt] m.north west)
  {and so on \ldots};

\end{tikzpicture}
}
    \caption{\texttt{PredSub} procedure: in Step 1 we select the subsample (red border), in Step 2 we apply ASE on the submatrix, and in Step 3 we estimate each remaining latent position from its cross-edge vector to the sampled vertices (blue border).}
    \label{fig:placeholder}
\end{figure}

\begin{remark}
\label{remark:localglobal}
\textbf{Statistical intuition.}
    Equation~\eqref{eq:cross_edge_identity} explains the statistical intuition behind the predictive step. The distribution of the cross-edge vector $A_{i,S}$ is governed by two components: the \textit{local} latent position $x_i$, which is specific to the $i^{th}$ vertex, and the \textit{global} latent structure of the subgraph, represented by $X_S$. Estimating both components from the single vector $A_{i,S}$ is generally infeasible. The key idea of \texttt{PredSub} is to estimate the global component $X_S$ first, using the sampled subgraph $A_S$ in Step 2. Once $X_S$ has been estimated, Equation~\eqref{eq:predsub} reduces the estimation of the local parameter $x_i$ to a vector operation involving  $A_{i,S}$. Thus, the sampled subgraph serves as an inferential anchor, and the remaining vertices can be incorporated one at a time in parallel through out-of-sample prediction.
\end{remark}

\begin{remark}
\label{remark:d}
\textbf{Choice of $d$.}
Note that our method requires $d$ as an input parameter. We assume that $d$ is known in our theoretical results. If $d$ is unknown, it can be estimated consistently via existing methods \citep[see][for more details]{chatterjee2015matrix,han2019universal,chakrabarty2025network}. For simulation studies, we specify $d$ according to the data generation process. Meanwhile, in real data analysis, we report results for multiple choices of $d$ to demonstrate that our method consistently aligns with the results obtained via ASE in substantially less time for any fixed embedding dimension. Although one may estimate $d$ and then apply \texttt{PredSub}, developing or evaluating such dimension selection strategies lies outside the scope of this work.
\end{remark}

\begin{remark}
\label{remark:compcost}
\textbf{Computational complexity.}
    Recall that the computational complexity of ASE on the full network is $\order(n^3)$ using standard spectral decomposition and $\order(n^2d)$ for randomized SVD.
    In comparison, the computational complexity of the \texttt{PredSub} procedure is $\order(m^3 +nmd)$ using standard spectral decomposition and $\order(m^2d + nmd)$ using randomized SVD. When $d \ll m \ll n$, the complexity becomes $\order(nmd)$, which means \texttt{PredSub} is $n/m$ times faster than ASE on the full network. When using the standard SVD, the computational advantage gets even bigger to $n^2/m^2$ times. 
\end{remark}

 \subsection{Scalable two-sample hypothesis testing}
\label{subsec:predsub-testing}

We now develop scalable two-sample testing procedures that leverage \texttt{PredSub}. Suppose we observe two independent networks on a common set of $n$ vertices, with adjacency matrices $A^{(1)}$ and $A^{(2)}$, generated from underlying probability matrices $P^{(1)}$ and $P^{(2)}$ under the GRDPG model. Our objective is to test
\begin{align}
H_0: P^{(1)}=P^{(2)}
\quad \text{vs.}\quad
H_1: P^{(1)}\ne P^{(2)} .
\label{hypotheses}
\end{align}

This testing problem arises in various scientific domains such as 
neurobiology and genetics \citep{ginestet2017hypothesis,zhang2009differential}, and there has been substantial methodological work on this problem in recent years \citep{tang2017semiparametric,levin2017central,ghoshdastidar2020two,bhadra2025bootstrap}.

We consider the method recently proposed by \cite{bhadra2025bootstrap} as the baseline full-sample procedure under the GRDPG framework.
In this method, we first apply ASE to the full adjacency matrices $A^{(1)}$ and $A^{(2)}$ to obtain the estimates $\hat P^{(1)}$ and $\hat P^{(2)}$. The observed test statistic is then taken to be
$$T_0 = \|\hat P^{(1)}-\hat P^{(2)}\|_F .$$
Under $H_0$, the common edge-probability matrix is estimated by the pooled estimator
$\hat P_0 = \frac{\hat P^{(1)}+\hat P^{(2)}}{2}$.
Then, for each bootstrap replicate $b=1,\ldots,B$, two independent bootstrap adjacency matrices $A^{\star{1,b}}$ and $A^{\star{2,b}}$ are generated from $\hat P_0$. Next, we apply ASE to $A^{\star{1,b}}$ and $A^{\star{2,b}}$ to obtain the bootstrapped estimates $\hat P^{\star{1,b}}$ and $\hat P^{\star{2,b}}$, and compute the bootstrap statistic
$T_b^\star = \|\hat P^{\star{1,b}}-\hat P^{\star{2,b}}\|_F$. The final $p$-value is obtained by comparing the observed statistic $T_0$ to the empirical distribution of $T_b^\star$.

The computational bottleneck of this full-network procedure arises from the repeated estimation of the full $n\times n$ probability matrices. Exact ASE requires $\order(n^3)$ time, and even iterative or randomized SVD-based implementations require approximately $\order(n^2d)$ time per bootstrap iteration. Since the bootstrap procedure repeats this estimation step $B$ times for two networks, the full-network test has computational complexity $\order(n^2dB)$ using iterative spectral methods, or $\order(n^3B)$ using exact eigendecomposition. In addition, each bootstrap replicate requires generating and storing full $n\times n$ adjacency matrices. Thus, the full-network bootstrap test becomes computationally prohibitive for large networks.

To address this computational bottleneck, we propose two scalable alternatives. The first, which we call \texttt{PredSubTest}, retains the Frobenius norm test statistic used in \cite{bhadra2025bootstrap}, but replaces the full-sample ASE estimators $\hat{P}^{(1)}$ and $\hat{P}^{(2)}$ with the \texttt{PredSub} estimator proposed in Section \ref{sec:estimation}. 
The second scalable alternative, which we call \texttt{PureSubTest}, goes one step further by replacing the full-network test statistic itself with a randomized subsampled statistic computed only on a selected subgraph. These two procedures represent different points on the trade-off between computational efficiency and statistical accuracy: \texttt{PredSubTest} provides a faster procedure for testing the original full-sample hypothesis, whereas \texttt{PureSubTest} tests a randomly sampled subgraph version of the null hypothesis and is therefore even more computationally efficient but statistically less accurate.

We first describe \texttt{PredSubTest}. The observed test statistic is
$$T_0 = \|\hat{P}^{(1)}_{PS} - \hat{P}^{(2)}_{PS}\|_F,$$
where $\hat{P}^{(1)}_{PS}$ and $\hat{P}^{(2)}_{PS}$ are obtained by applying Algorithm~\ref{algo1} on $A^{(1)}$ and $A^{(2)}$, respectively. Under the null hypothesis, we compute the pooled estimator
$\hat P_0 = \frac{\hat{P}^{(1)}_{PS}+ \hat{P}^{(2)}_{PS}}{2}$,
and generate bootstrap samples $A^{\star{1,b}}$ and $A^{\star{2,b}}$ from $\hat P_0$. For each bootstrap sample, the probability matrices are again estimated using \texttt{PredSub}, and the bootstrapped test statistic is computed using the Frobenius norm difference between the two \texttt{PredSub} estimates.
See Algorithm~\ref{algo2} for details.

\begin{algorithm}[ht]
\caption{Hypothesis Testing with Predictive Subsampling}
\label{algo2}
\begin{algorithmic}[1]

\Require Adjacency matrices $A^{(1)}_{n\times n}, A^{(2)}_{n\times n}$, subsample size $m$, rank $d$, sampling technique $\mathcal{S}$, number of bootstrap samples $B$
\Ensure P-value $p$

\Procedure{PredSubTest}{$A^{(1)}, A^{(2)}, m, d, \mathcal{S}, B$}

\State Select a subsample $S$ of $m$ vertices using subsampling technique $\mathcal{S}$ and estimate the probability matrices $\hat{P}^{(1)}_{PS}$ and $\hat{P}^{(2)}_{PS}$ using Algorithm~\ref{algo1}

\State Compute the test statistic $T_0 = \|\hat{P}^{(1)}_{PS} - \hat{P}^{(2)}_{PS}\|_F$

\State Under the null hypothesis, compute the pooled estimate $\hat{P}_0 = (\hat{P}^{(1)}_{PS} + \hat{P}^{(2)}_{PS})/2$

\For{$b = 1$ to $B$}

    \State Generate bootstrap samples $A^\star_{1,b}$ and $A^\star_{2,b}$ where for any $j \in \{1,2\}$, the entries of $A^\star_{j,b}$ are independent Bernoulli random variables with success probabilities given by the corresponding entries in $\hat{P}_0$.

    \State Estimate the probability matrices $\hat{P}^\star_{1,b}$ and $\hat{P}^\star_{2,b}$ using Algorithm~\ref{algo1}

    \State Compute $T^\star_b = \|\hat{P}^\star_{1,b} - \hat{P}^\star_{2,b}\|_F$

\EndFor

\State Output the p-value $p = \frac{1}{B} \sum_{b=1}^{B} \mathbb{I}(T^\star_b > T_0)$

\EndProcedure

\end{algorithmic}
\end{algorithm}

We now discuss the computational complexity of Algorithm~\ref{algo2}. The initial \texttt{PredSub} estimation step can be carried out in $\order(nmd)$ time. The Frobenius norm statistic can be computed in $\order(nd^2)$ time by exploiting the low-rank structure of $\hat P_{PS}$. For each bootstrap sample, we only need to generate the entries corresponding to the $n\times m$ cross-edge matrix with endpoints in the selected subsample $S$, rather than generating the full $n\times n$ adjacency matrix. Thus, the bootstrap data generation and estimation steps require $\order(nmdB)$ time, while the computation of the bootstrap test statistics requires $\order(nd^2B)$ time. In summary, Algorithm~\ref{algo2} has a total computational complexity $\order(nmdB)$, which makes it $n/m$ times faster than the full-sample test.
Recall that $m \ll n$ which makes this speedup quite substantial.

We next describe the \texttt{PureSubTest} procedure. Note that, under the null hypothesis, $P_S^{(1)}=P_S^{(2)}$ for any subset $S\subset V$. Recall that in Steps 1 and 2 of \texttt{PredSub}, we first select a subsample $S\subset V$ and then estimate the corresponding subgraph probability matrix $\hat{P}_S$. This motivates the subsampled test statistic
$$T_0^S = \|\hat{P}_S^{(1)} - \hat{P}_S^{(2)}\|_F,$$
where $\hat{P}_S^{(1)}$ and $\hat{P}_S^{(2)}$ are computed by applying only Steps 1 and 2 of \texttt{PredSub} to $A^{(1)}$ and $A^{(2)}$, respectively. In other words, Step 3 of \texttt{PredSub}, which predicts the latent positions of the remaining vertices from the cross-edge vectors, is not used. The complete procedure is given in Algorithm~\ref{algo3}. 

Compared to \texttt{PredSubTest}, which uses all $n$ vertices, \texttt{PureSubTest} uses only the $m$ subsampled vertices. Therefore, it requires only $\order(m^2dB)$ time when iterative spectral methods are used, making it $n^2/m^2$ times faster than the full network method and $n/m$ times faster than \texttt{PredSubTest}.

\begin{algorithm}
\caption{Hypothesis Testing with Single Subgraph}
\label{algo3}
\begin{algorithmic}[1]

\Require Adjacency matrices $A^{(1)}, A^{(2)}$, subsample size $m$, rank $d$, sampling technique $\mathcal{S}$
\Ensure P-value $p$

\Procedure{PureSubTest}{$A^{(1)}, A^{(2)}, m, d, \mathcal{S}$}

\State Choose a subsample $S \subset V$ of size $m$ according to sampling algorithm $\mathcal{S}$

\State Estimate the probability submatrices $\hat{P}_S^{(1)}$ and $\hat{P}_S^{(2)}$ using ASE on $A_S^{(1)}$ and $A_S^{(2)}$, respectively

\State Compute the test statistic $T_0^S = \|\hat{P}_S^{(1)} - \hat{P}_S^{(2)}\|_F$

\State Under the null hypothesis, compute the pooled estimate $\hat{P}_0^S = (\hat{P}_S^{(1)} + \hat{P}_S^{(2)})/2$

\For{$b = 1$ to $B$}

    \State Generate bootstrap samples $A^S_{1,b}$ and $A^S_{2,b}$  where for any $j \in \{1,2\}$, the entries of $A^S_{j,b}$ are independent Bernoulli random variables with success probabilities given by the corresponding entries in $\hat{P}_0^S$.

    \State Estimate the probability matrices $\hat{P}^S_{1,b}$ and $\hat{P}^S_{2,b}$ using ASE

    \State Compute $T^S_b = \|\hat{P}^S_{1,b} - \hat{P}^S_{2,b}\|_F$

\EndFor

\State Output the p-value $p = \frac{1}{B} \sum_{b=1}^{B} \mathbb{I}(T^S_b > T_0^S)$

\EndProcedure

\end{algorithmic}
\end{algorithm}

\begin{remark}
    \textbf{\texttt{PredSubTest} vs. \texttt{PureSubTest}.} While Algorithm~\ref{algo3} is faster than Algorithm~\ref{algo2},  the tradeoff is that the test statistic is now randomized through the choice of $S$. Consequently, \texttt{PureSubTest} may lose power if the selected subgraph does not capture the difference between $P^{(1)}$ and $P^{(2)}$. However, as we demonstrate in subsequent sections, both theoretical analysis and empirical results indicate that if $P^{(1)}$ and $P^{(2)}$ are sufficiently different, then a uniformly chosen random subsample is still likely to capture this difference, resulting in a test procedure with high power. In summary, a proper choice of $m$ in Algorithm~\ref{algo3} allows us to achieve the right balance between computational efficiency and statistical power.
Please see Remarks \ref{remark:PredSubTest},\ref{remark:PureSubTest}, and \ref{remark:m} for a formal characterization of these trade-offs.
\end{remark}

\begin{remark}
\label{remark:norm}
\textbf{Choice of norm.}
In this section, we have presented \texttt{PredSubTest} and \texttt{PureSubTest} using test statistics based on the Frobenius norm.
This is because full-sample tests based on the Frobenius norm are standard in the current literature \citep{tang2017semiparametric,levin2017central, bhadra2025bootstrap}. The proposed scalable testing procedures, however, are not tied to this particular choice of matrix norm.
One particularly useful alternative is the two-to-infinity norm. While the Frobenius norm aggregates differences over all entries of the probability matrices and is therefore well suited to detecting global alternatives, the two-to-infinity norm measures the largest rowwise difference and therefore works better with {local} alternatives. Also note that the ASE estimator enjoys strong concentration properties with respect to the two-to-infinity norm \citep{xie2024entrywise}. Accordingly, in Section~\ref{sec4} we establish consistency for both Frobenius norm and two-to-infinity-norm versions of \texttt{PredSubTest} and \texttt{PureSubTest}. Please also see Remark \ref{remark:PredSubTest} in Section~\ref{sec4} for more details.
\end{remark}

\section{Theoretical Results}
\label{sec4}
In this section, we establish the theoretical properties of \texttt{PredSub}, \texttt{PredSubTest}, and \texttt{PureSubTest}.
Recall that $A \thicksim \text{GRDPG}_{p,q}(P)$ as in Equation \eqref{eq:grdpg}. For a subset $S$ of $m$ vertices, we have $S^c = [n]\backslash S$. Now we define the following quantities.

\begin{itemize}
    \item $A_S$ and $P_S$ are the submatrices of $A$ and $P$, respectively, with rows and columns indexed by $S$. 
    \item $X_S = U_{P_S} |D_{P_S}|^{1/2}$ is the ASE of $P_S$.
    \item $X_{S^c} = P_{S^c,S} X_{S}(X_S^\top X_S)^{-1} I_{p,q}$ is the predictive embedding for the rows in $S^c$.
    \item $X_{PS} = [X_S^{\top} \mid X_{S^c}^{\top}]^{\top}$.
    \item $Z = U_{P} |D_{P}|^{1/2}$ is the ASE of $P$.
    \item 
    $X_{PS} = ZQ$ 
    where 
    $Q = I_{p,q} Z_S^{\top} X_S (X_S^{\top} X_S)^{-1} I_{p,q}\in \mathbb{O}(p,q)$ is {\em indefinite orthogonal}.
    \item $\hat{X}_S = U_{A_S} |D_{A_S}|^{1/2}$ is the ASE of $A_S$.
    \item $\hat{X}_{S^c} = A_{S^c,S} \hat{X}_S(\hat{X}_S^\top \hat{X}_S)^{-1} I_{p,q}$ is the predictive embedding for the rows in $S^c$ as in Section~\ref{sec:estimation}.
    \item $\hat{X}_{PS} = [\hat{X}_S^{\top}\mid \hat{X}_{S^c}^{\top}]^\top$ be the estimate obtained through Algorithm~\ref{algo1}.
\end{itemize} 
Observe that $X_{PS}$ is {\em random} as $X_S$ depends on $S$, and $X_{S^c}$ then depends on $X_S$. Nevertheless, for any fixed realization of $S$, the underlying probability matrix can be completely characterized by $P = \mathfrak{S} X_{PS}I_{p,q}X_{PS}^\top \mathfrak{S}^{\top}$ for some permutation matrix $\mathfrak{S}$ (dependent on $S$).  
Hence, estimating $X_{PS}$ is enough to get an estimator of $P$. Our main results therefore focus on estimating $X_{PS}$. 
Finally, by definition, our estimator should satisfy $\hat{X}_{PS}W \approx X_{PS}$ for some {\em orthogonal} matrix $W$. In this section, we aim to bound the estimation error $\min_{W\in \mathbb{O}(d)}\|\hat{X}_{PS}W - X_{PS}\|_{2\to\infty}$. For our theoretical results, we assume that $p+q \leq d$ for some finite constant $d$ not depending on $n$.

\subsection{Estimation consistency of \texttt{PredSub}}
\label{sec:estimation-theory}
We assume that the following regularity conditions hold. 

\begin{assumption}
    $m = \Omega((\log n)^{1+a})$ for some $a>0$ \label{ass1}
\end{assumption}



\begin{assumption}
    $P$ is an uniformly sparse matrix i.e., there exists fixed constants $c_1 >0 $ and $c_2 \leq 1$, and a sequence $\{\rho_n\}_{n=1}^\infty$ with $\rho_n \leq 1$ for all $n$ such that $c_1 \rho_n \leq P_{ij} \leq c_2\rho_n$ for all $n \in \N$ and all $1 \leq i \leq j \leq n$. \label{ass2}
\end{assumption}

 \begin{assumption}
     $\kappa(P) \leq c_0$ for some $c_0>0$. \label{ass3}
 \end{assumption}

 \begin{assumption}
$m\rho_n = \Omega(\log m)$ \label{ass4}
 \end{assumption}

We now discuss the motivations behind Assumptions \ref{ass1} through \ref{ass4}. Assumption~\ref{ass1} imposes a lower bound on the subsample size $m$ relative to the network size $n$. 
Assumption~\ref{ass2} formalizes the notion of uniformly sparse graphs and ensures that all entries in the probability matrix $P$ are of the same order, governed by the sparsity parameter $\rho_n$; see Definition~13 in \cite{lyzinski2016community} for a similar condition. Assumption~\ref{ass3} states that $P$ has a bounded condition number and is thus similar  to the eigenvalue gap condition in Assumption 2.1 of \cite{tang2017semiparametric}. Note that Assumption~\ref{ass3} is mainly for ease of exposition as otherwise we can simply include factors depending on $\kappa(P)$ throughout our theoretical results. Finally, Assumption~\ref{ass4} requires the expected degree in the sampled subgraph to be at least of order $\Omega(\log m)$, which in turn imposes a lower bound on the sparsity parameter $\rho_n$, given the network size and the subsample size. We note that Assumption~\ref{ass4} is slightly more restrictive than the condition $n \rho_n = \Omega(\log n)$ typically encountered in the literature (e.g., Assumption~2 of \cite{xie2024entrywise}), but this is expected as the first step in our estimation procedure uses the subsampled network $A_S$ and not the full network $A$. 

The following two theorems establish upper bounds on the estimation error $\hat{X}_{PS} - X_{PS}$, in two-to-infinity and Frobenius norms. In particular these bounds, while properly scaled, all converge to zero as $n$ tends to infinity, thereby implying consistency of $\hat{X}_{PS}$. For ease of presentation, we have expressed our bounds in the form of $Z = \mathcal{O}_{p}(g(n))$ for some random variable $Z$ to mean that, for any constant $c > 0$, there exists an integer $n_0$ and a constant $C > 0$ (both possibly depending on $c$) such that for all $n\geq n_0, |Z|\leq Cg(n)$ with probability at least $1-m^{-c}$. 
 \begin{theorem}\label{thm1}
     Let $A \thicksim \text{GRDPG}_{p,q}(P)$ as in Equation \eqref{eq:grdpg}
     and $\hat{X}_{PS}$ be the estimate of $X_{PS}$ as obtained in Algorithm \ref{algo1}. Given Assumptions \ref{ass1}-\ref{ass4} hold, there exists an orthogonal transformation $W \in \mathbb{O}(d)$ such that
     $$ \|\hat{X}_{PS}W - X_{PS} \|_{2\to\infty} = \order_p\left(\sqrt{\dfrac{\log n}{m}}\right).$$
 \end{theorem}
\begin{theorem}\label{thm2}
    Consider the same setting as in Theorem \ref{thm1}. Given Assumptions \ref{ass1}-\ref{ass4} hold, there exists an orthogonal transformation $W \in \mathbb{O}(d)$ such that
     $$ \|\hat{X}_{PS}W-X_{PS}\|_F = \order_p\left(\sqrt{\frac{n}{m}}\right).$$
\end{theorem}

Theorems~\ref{thm1} and \ref{thm2} imply that the estimate $\hat{P}_{PS}$ given by Algorithm~\ref{algo1} is also consistent. This is formalized in the following corollary.
\begin{corollary}
Consider the same setting as in Theorem~\ref{thm1} and let $\hat{P}_{PS} = \hat{X}_{PS}I_{p,q}\hat{X}_{PS}^\top$ be the estimate for $P$ using Algorithm~\ref{algo1}. Then 
    \begin{gather*}
       \|\mathfrak{S} \hat{P}_{PS} \mathfrak{S}^{\top} - P \|_F= \order_p\bigl(\|X_{PS}\|_{F} \sqrt{n/m}\bigr), \\
    \|\mathfrak{S} \hat{P}_{PS} \mathfrak{S}^{\top} - P \|_{2\to\infty}=\order_p\Bigl(\|X_{PS}\|_{2 \to \infty} \sqrt{\frac{n \log n}{m}}\Bigr).
    \end{gather*}
    for some $n \times n$ permutation matrix $\mathfrak{S}$ dependent only on $S$. 
     \label{corollary1}
\end{corollary}

Now, based on these theorems and corollaries, we move forward to establish the theoretical guarantees for the proposed two-sample testing procedures. 

\subsection{Hypothesis testing consistency of \texttt{PredSubTest} and \texttt{PureSubTest}}

The following two theorems establish the hypothesis testing consistency of \texttt{PredSubTest} in terms of the Frobenius norm and the two-to-infinity norm.

\begin{theorem}\label{testing0}
    Let $\hat{P}_{PS}^{(1)}$ and $\hat{P}_{PS}^{(2)}$ be estimates of $P^{(1)}$ and $P^{(2)}$ obtained using Algorithm \ref{algo1}, respectively. 
    Define the test statistic $T_F = \|\hat{P}_{PS}^{(1)}-\hat{P}_{PS}^{(2)}\|_F$ and let $R_{F} = \|\hat{X}_{PS}^{(1)}\|_{F} + \|\hat{X}_{PS}^{(2)}\|_{F}$. There exists a universal constant $K_4$ such that if $\mathcal{R}_n$ is the rejection region of the form 
    $\{ T_n > 2 K_4 R_{F} \sqrt{n/m}\}$ then for sufficiently large $n$ the following holds.
    \begin{itemize}
    \item If $P^{(1)} = P^{(2)}$ then $\mathbb{P}(T_F  \in \mathcal{R}_n) = o_{p}(1)$. 
    \item If $\|P^{(1)} - P^{(2)}\|_{F} = \Omega(n \sqrt{\rho_n/m})$ then 
    $\mathbb{P}(T_F \not \in \mathcal{R}_n) = o_{p}(1)$. 
    \end{itemize}    
\end{theorem}

\begin{theorem}\label{testing1}
 Let $\hat{P}_{PS}^{(1)}$ and $\hat{P}_{PS}^{(2)}$ be estimates of $P^{(1)}$ and $P^{(2)}$ obtained using Algorithm \ref{algo1}, respectively. 
    Define the test statistic $T_{2 \to\infty}= \|\hat{P}_{PS}^{(1)}-\hat{P}_{PS}^{(2)}\|_{2 \to \infty}$ and let $R_{2 \to \infty} = \|\hat{X}_{PS}^{(1)}\|_{2 \to \infty} + \|\hat{X}_{PS}^{(2)}\|_{2 \to \infty}$. There exists a universal constant $K_4$ such that if $\tilde{\mathcal{R}}_n$ is the rejection region of the form $\{T_{2 \to \infty} \geq 2 K_4 R_{2 \to \infty} \sqrt{(n \log n)/m}$ then for sufficiently large $n$ the following holds.
        \begin{itemize}
    \item If $P^{(1)} = P^{(2)}$ then $\mathbb{P}(T_{2 \to \infty} \in \tilde{\mathcal{R}}_n) = o_{p}(1)$. 
    \item If $\|P^{(1)} - P^{(2)}\|_{2 \to \infty} = \Omega(\sqrt{(n \rho_n \log n)/m})$ then 
    $\mathbb{P}(T_{2 \to \infty} \not \in \tilde{\mathcal{R}}_n) = o_{p}(1)$. 
    \end{itemize}    
\end{theorem}

\begin{remark}
\label{remark:PredSubTest}
\textbf{Choice of norm.}
Theorems~\ref{testing0} and \ref{testing1} imply that the \texttt{PredSubTest} procedures based on $T_F$ and $T_{2 \to \infty}$ are both consistent, provided that $P^{(1)}$ is sufficiently different from $P^{(2)}$. However, the class of alternatives for which these two versions are consistent are quite different. More specifically, suppose that the entrywise differences between $P^{(1)}$ and $P^{(2)}$, if non-zero, are all of order $\rho_n$ and let 
$\gamma_n$ denote the proportion of entries that are different. Then, \texttt{PredSubTest} based on $T_F$ is consistent when $\|P^{(1)} - P^{(2)}\|_{F} = \Omega(n \sqrt{\rho_n/m})$, or equivalently $\gamma_n = \Omega(1/(m \rho_n))$. 
In contrast, to ensure consistency of \texttt{PredSubTest} based on $T_{2 \to \infty}$ we only need one row of $P^{(1)}$ to be significantly different from the corresponding row of $P^{(2)}$, and this can be satisfied even under the much weaker condition $\gamma_n = \Omega((\log n)/(n m \rho_n))$. 
The test procedure based on $T_{2 \to \infty}$ can therefore identify smaller or more localized differences between $P^{(1)}$ and $P^{(2)}$ compared to that for $T_F$. 
%
\end{remark}

The following two theorems establish the consistency of \texttt{PureSubTest} in terms of the the two-to-infinity norm and the Frobenius norm, respectively.

\begin{theorem}\label{testing2}
   Let $\hat{P}_{S}^{(1)},\hat{P}_{S}^{(2)}$ be the ASE estimate for a chosen subsample $S$ of size $m$ from both the graphs respectively. 
    Define the test statistic $T_{2 \to\infty}^S= \|\hat{P}_{S}^{(1)}-\hat{P}_{S}^{(2)}\|_{2 \to \infty}$ and let $R_{2 \to \infty}^S = \|\hat{X}_{S}^{(1)}\|_{2 \to \infty} + \|\hat{X}_{S}^{(2)}\|_{2 \to \infty}$. There exists a universal constant $C_1$ such that if $\Bar{\mathcal{R}}_n$ is the rejection region of the form $\{T_{2 \to \infty}^S \geq 2 C_1 R_{2 \to \infty}^S \sqrt{\log m}\}$ then for sufficiently large $n$ the following holds.
        \begin{itemize}
    \item If $P^{(1)} = P^{(2)}$ then $\mathbb{P}(T_{2 \to \infty}^S \in \Bar{\mathcal{R}}_n) = o_{p}(1)$. 
    \item If there exists at least $k = \Omega((n\log n)/m)$ many rows in $P^{(1)}$ and $P^{(2)}$  with $\|P_i^{(1)}-P_i^{(2)}\|_2=\Omega(\sqrt{(n\rho_n \log m)/m})$ 
    then 
    $\mathbb{P}(T_{2 \to \infty}^S \not \in \Bar{\mathcal{R}}_n) = o_{p}(1)$. 
    \end{itemize}
\end{theorem}

\begin{theorem}\label{thm:puresubF}
    Let $\hat{P}_{S}^{(1)},\hat{P}_{S}^{(2)}$ be the ASE estimate for a chosen subsample $S$ of size $m$ from both the graphs respectively. 
    Define the test statistic $T_F^S= \|\hat{P}_{S}^{(1)}-\hat{P}_{S}^{(2)}\|_F$ and let $R_F^S = \|\hat{X}_{S}^{(1)}\|_F + \|\hat{X}_{S}^{(2)}\|_F$. There exists a universal constant $C_1$ such that if $\Ddot{\mathcal{R}}_n$ is the rejection region of the form $\{T_F^S \geq 2 C_1 R_F^S\}$ then for sufficiently large $n$ the following holds.
    \begin{itemize}
    \item If $P^{(1)} = P^{(2)}$ then $\mathbb{P}(T_F^S  \in \Ddot{\mathcal{R}}_n) = o_{p}(1)$. 
    \item If $\|P^{(1)} - P^{(2)}\|_{F} = \Omega(n^2\sqrt{\rho_n/m^3})$ then 
    $\mathbb{P}(T_F^S \not \in \Ddot{\mathcal{R}}_n) = o_{p}(1)$. 
    \end{itemize}    
\end{theorem}

\begin{remark}
\label{remark:PureSubTest}
\textbf{\texttt{PredSubTest} vs. \texttt{PureSubTest}.}
    Note that \texttt{PureSubTest} requires a sufficiently large number of rows to differ between the two matrices in order to reliably detect a difference, whereas \texttt{PredSubTest} can, in principle, detect alternatives where only a single row differs. Furthermore the magnitudes of the detectable differences are identical, up to logarithmic factor, across the two approaches, 
    This suggests that \texttt{PureSubTest} may be more effective under strong alternatives, while \texttt{PredSubTest} offers greater flexibility and robustness across a broader range of settings.
\end{remark}

\begin{remark}
\label{remark:m}
\textbf{Choice of $m$.}
   Across all tests, increasing the subsample size $m$ expands the class of alternatives for which the procedure is consistent, thereby yielding higher power. Conversely, smaller values of $m$ produce a more computationally efficient but less powerful test that can only detect relatively large differences between the two probability matrices. This highlights an inherent trade-off between the magnitude of the underlying difference and the choice of $m$. Specifically, for given values of $\rho_n$ and the proportion of differing entries $\gamma_n$, the subsample size $m$ must be chosen appropriately to achieve the desired power. When the differences between the matrices are small, a larger subsample size is necessary to guarantee consistency of the tests. 
\end{remark}



\section{Simulation Study}
\label{sec5}
In this section, we empirically evaluate the computational and statistical performance of the proposed estimation and testing methods using simulated data.

\subsection{Estimation under GRDPG}
\label{sec4.1}
We quantify estimator performance as $\frac{\|\hat{P} - P\|_F}{\|P\|_F}$, i.e., the relative error norm difference between the true probability matrix and its estimate. The relative error provides a form of standardization, allowing us to compare results across different parameter settings and replications. There are three global parameters, namely $n$ (number of vertices), $d$ (rank of $P$), and $\rho_n$ (the sparsity parameter). Given these parameters, we constructed $P$ and $A$  under the mixed membership blockmodel of \cite{Airoldi2008}, which is a special case of GRDPG, as follows:
\begin{enumerate}
    \item Generate a $d\times d$ matrix $B$ such that $B_{ij} = 0.5$ for $i\ne j$ and $B_{ii}\thicksim \text{Uniform}(0,1)$.
    \item Generate an $n\times d$ matrix $\Pi$ with rows $\{\pi_i\}_{i=1}^{n}$ drawn independently from 
    $\mathrm{Dirichlet}(0.5, \dots, 0.5)$.
    \item Compute $P = \rho_n \Pi B\Pi^\top$ 
    \item Generate $A \thicksim \text{Bernoulli}(P)$
\end{enumerate}
For our simulation experiments, we used $n \in \{6 \times 10^4,10^5, 1.5 \times 10^5\}$, $d \in \{5,10,20\}$,
and $\rho_n \in \{0.01, 0.04\}$. 
For each combination of $n$, $d$, and $\rho_n$, we obtained a single probability matrix $P$ as described above and generated 1000 independent networks from $P$. 
Finally, the subsample size for \texttt{PredSub} is given by 
$m = \lceil(\log n)^{4+b}\rceil$ for some $b \in (-0.375,0.375)$. 
Note that the choices of $m$ here align with our Assumption~\ref{ass1} and can also be represented as $\lceil(\log n)^{1+a}\rceil$ for some $a>0$.

Under this setup, we compared \texttt{PredSub} with the full-sample benchmark, which we refer to as ASE. The scalable estimator proposed by \cite{chakraborty2025scalable} was not included in this comparison since it was developed specifically under the RDPG framework. A separate simulation study comparing \texttt{PredSub} with their method is presented in Section \ref{sec:estimation_RDPG}.
The runtimes and relative errors are summarized in Tables~\ref{tab1}, \ref{tab2} and \ref{tab3} for $n=6 \times 10^4,10^5, 1.5 \times 10^5,$ respectively.

\begin{table}[ht]
\centering
\small
\setlength{\tabcolsep}{3pt}
\sisetup{table-format=3.2}
\begin{tabular}{
  c
  l
  @{\hspace{1.5em}}S[table-format=1.3]
  S[table-format=3.1]
  S[table-format=6.2]
  S[table-format=1.2]
  S[table-format=1.2]
  S[table-format=1.2]
  S[table-format=3.1]
  S[table-format=6.2]
  S[table-format=1.2]
  S[table-format=1.2]
  S[table-format=1.2]
}
\toprule
& & &
  \multicolumn{5}{c}{$\rho_n = 0.01$} &
  \multicolumn{5}{c}{$\rho_n = 0.04$} \\
\cmidrule(lr){4-8} \cmidrule(lr){9-13}
& & &
  \multicolumn{2}{c}{Time (secs)} &
  \multicolumn{3}{c}{Accuracy} &
  \multicolumn{2}{c}{Time (secs)} &
  \multicolumn{3}{c}{Accuracy} \\
\cmidrule(lr){4-5} \cmidrule(lr){6-8} \cmidrule(lr){9-10} \cmidrule(lr){11-13}
$d$ & Method & {$b$} &
  {Mean} & {$\downarrow$ times} & {Error} & {SD} & {$\uparrow$ times} &
  {Mean} & {$\downarrow$ times} & {Error} & {SD} & {$\uparrow$ times} \\
\midrule
5  & ASE     & {--}   & 35.2 & {--}   & 0.25 & 0.00 & {--}  & 115.1 & {--}   & 0.13 & 0.00 & {--}  \\
   & \texttt{PredSub} & -0.375 &  0.3 & 134.10 & 0.35 & 0.00 & 1.42 &   0.7 & 160.68 & 0.18 & 0.00 & 1.40 \\
   &         & -0.125 &  0.7 &  53.86 & 0.29 & 0.00 & 1.16 &   2.2 &  51.84 & 0.15 & 0.00 & 1.15 \\
   &         &  0.125 &  2.5 &  14.31 & 0.25 & 0.00 & 1.00 &   8.0 &  14.41 & 0.13 & 0.00 & 1.00 \\
\midrule
10 & ASE     & {--}   & 42.3 & {--}   & 0.35 & 0.00 & {--}  & 136.5 & {--}   & 0.18 & 0.00 & {--}  \\
   & \texttt{PredSub} & -0.375 &  0.3 & 137.18 & 0.44 & 0.00 & 1.23 &   0.9 & 153.04 & 0.22 & 0.00 & 1.24 \\
   &         & -0.125 &  0.8 &  51.75 & 0.36 & 0.00 & 1.03 &   2.8 &  48.87 & 0.18 & 0.00 & 1.04 \\
   &         &  0.125 &  3.2 &  13.40 & 0.33 & 0.00 & 0.93 &   9.9 &  13.78 & 0.16 & 0.00 & 0.93 \\
\midrule
20 & ASE     & {--}   & 66.9 & {--}   & 0.51 & 0.00 & {--}  & 205.9 & {--}   & 0.25 & 0.00 & {--}  \\
   & \texttt{PredSub} & -0.375 &  0.5 & 140.94 & 0.57 & 0.00 & 1.14 &   1.3 & 158.70 & 0.29 & 0.00 & 1.14 \\
   &         & -0.125 &  1.3 &  51.63 & 0.49 & 0.00 & 0.97 &   4.1 &  50.21 & 0.24 & 0.00 & 0.97 \\
   &         &  0.125 &  4.7 &  14.17 & 0.46 & 0.00 & 0.90 &  14.5 &  14.18 & 0.23 & 0.00 & 0.90 \\
\bottomrule
\end{tabular}
\caption{Time (in secs) and error comparison between ASE and \texttt{PredSub} with
  $m = \lceil(\log n)^{4+b}\rceil$ for $n=60000$ and various $(d,\rho_n)$
  values under GRDPG setting.}
\label{tab1}
\end{table}

\begin{table}[ht]
\centering
\small
\setlength{\tabcolsep}{3pt}
\sisetup{table-format=3.2}

\begin{tabular}{
  c
  l
  @{\hspace{1.5em}}S[table-format=1.2]
  S[table-format=3.1]
  S[table-format=6.2]
  S[table-format=1.2]
  S[table-format=1.2]
  S[table-format=1.2]
  S[table-format=3.1]
  S[table-format=6.2]
  S[table-format=1.2]
  S[table-format=1.2]
  S[table-format=1.2]
}
\toprule
& & &
  \multicolumn{5}{c}{$\rho_n = 0.01$} &
  \multicolumn{5}{c}{$\rho_n = 0.04$} \\
\cmidrule(lr){4-8} \cmidrule(lr){9-13}
& & &
  \multicolumn{2}{c}{Time (secs)} &
  \multicolumn{3}{c}{Accuracy} &
  \multicolumn{2}{c}{Time (secs)} &
  \multicolumn{3}{c}{Accuracy} \\
\cmidrule(lr){4-5} \cmidrule(lr){6-8} \cmidrule(lr){9-10} \cmidrule(lr){11-13}
$d$ & Method & {$b$} &
  {Mean} & {$\downarrow$ times} & {Error} & {SD} & {$\uparrow$ times} &
  {Mean} & {$\downarrow$ times} & {Error} & {SD} & {$\uparrow$ times} \\
\midrule
5  & ASE     & {--}  & 155.1 & {--}   & 0.19 & 0.00 & {--}  & 434.6 & {--}   & 0.10 & 0.00 & {--}  \\
   & \texttt{PredSub} & -0.25 &   0.7 & 233.73 & 0.28 & 0.00 & 1.44 &   2.3 & 188.58 & 0.14 & 0.00 & 1.44 \\
   &         &  0.00 &   2.2 &  69.30 & 0.22 & 0.00 & 1.16 &   7.2 &  60.32 & 0.12 & 0.00 & 1.17 \\
   &         &  0.25 &   8.0 &  19.30 & 0.19 & 0.00 & 1.01 &  26.2 &  16.58 & 0.10 & 0.00 & 1.03 \\
\midrule
10 & ASE     & {--}  & 184.8 & {--}   & 0.27 & 0.00 & {--}  & 603.7 & {--}   & 0.14 & 0.00 & {--}  \\
   & \texttt{PredSub} & -0.25 &   0.8 & 227.33 & 0.34 & 0.00 & 1.25 &   2.9 & 206.46 & 0.17 & 0.00 & 1.25 \\
   &         &  0.00 &   2.9 &  64.66 & 0.28 & 0.00 & 1.04 &   9.0 &  66.91 & 0.14 & 0.00 & 1.04 \\
   &         &  0.25 &  10.0 &  18.41 & 0.26 & 0.00 & 0.94 &  31.9 &  18.94 & 0.13 & 0.00 & 0.94 \\
\midrule
20 & ASE     & {--}  & 298.0 & {--}   & 0.39 & 0.00 & {--}  & 931.7 & {--}   & 0.19 & 0.00 & {--}  \\
   & \texttt{PredSub} & -0.25 &   1.3 & 237.13 & 0.45 & 0.00 & 1.15 &   4.2 & 220.07 & 0.22 & 0.00 & 1.15 \\
   &         &  0.00 &   4.2 &  70.44 & 0.38 & 0.00 & 0.98 &  13.2 &  70.52 & 0.19 & 0.00 & 0.98 \\
   &         &  0.25 &  15.5 &  19.22 & 0.35 & 0.00 & 0.90 &  47.5 &  19.60 & 0.18 & 0.00 & 0.90 \\
\bottomrule
\end{tabular}

\caption{Time (in secs) and error comparison between ASE and \texttt{PredSub} with
$m = \lceil(\log n)^{4+b}\rceil$ for $n=100000$ and various $(d,\rho_n)$
values under GRDPG setting.}
\label{tab2}
\end{table}

\begin{table}[ht]
\centering
\small
\setlength{\tabcolsep}{3pt}
\sisetup{table-format=3.2}
\begin{tabular}{
  c
  l
  @{\hspace{1.5em}}S[table-format=1.3]
  S[table-format=3.1]
  S[table-format=6.2]
  S[table-format=1.2]
  S[table-format=1.2]
  S[table-format=1.2]
  S[table-format=3.1]
  S[table-format=6.2]
  S[table-format=1.2]
  S[table-format=1.2]
  S[table-format=1.2]
}
\toprule
& & &
  \multicolumn{5}{c}{$\rho_n = 0.005$} &
  \multicolumn{5}{c}{$\rho_n = 0.01$} \\
\cmidrule(lr){4-8} \cmidrule(lr){9-13}
& & &
  \multicolumn{2}{c}{Time (secs)} &
  \multicolumn{3}{c}{Accuracy} &
  \multicolumn{2}{c}{Time (secs)} &
  \multicolumn{3}{c}{Accuracy} \\
\cmidrule(lr){4-5} \cmidrule(lr){6-8} \cmidrule(lr){9-10} \cmidrule(lr){11-13}
$d$ & Method & {$b$} &
  {Mean} & {$\downarrow$ times} & {Error} & {SD} & {$\uparrow$ times} &
  {Mean} & {$\downarrow$ times} & {Error} & {SD} & {$\uparrow$ times} \\
\midrule
5  & ASE     & {--}   & 271.9 & {--}   & 0.22 & 0.00 & {--}  & 498.8 & {--}   & 0.16 & 0.00 & {--}  \\
   & \texttt{PredSub} & -0.125 &   1.0 & 260.56 & 0.32 & 0.00 & 1.42 &   1.9 & 255.94 & 0.22 & 0.00 & 1.42 \\
   &         &  0.125 &   4.0 &  67.31 & 0.26 & 0.00 & 1.15 &   6.3 &  79.27 & 0.18 & 0.00 & 1.15 \\
   &         &  0.375 &  13.9 &  19.51 & 0.22 & 0.00 & 1.00 &  23.7 &  21.08 & 0.16 & 0.00 & 1.00 \\
\midrule
10 & ASE     & {--}   & 311.4 & {--}   & 0.32 & 0.00 & {--}  & 561.7 & {--}   & 0.22 & 0.00 & {--}  \\
   & \texttt{PredSub} & -0.125 &   1.3 & 244.93 & 0.39 & 0.00 & 1.24 &   2.4 & 230.74 & 0.28 & 0.00 & 1.24 \\
   &         &  0.125 &   4.8 &  64.88 & 0.33 & 0.00 & 1.03 &   7.9 &  71.33 & 0.23 & 0.00 & 1.03 \\
   &         &  0.375 &  17.0 &  18.31 & 0.29 & 0.00 & 0.93 &  29.0 &  19.38 & 0.21 & 0.00 & 0.93 \\
\midrule
20 & ASE     & {--}   & 521.9 & {--}   & 0.45 & 0.00 & {--}  & 930.4 & {--}   & 0.32 & 0.00 & {--}  \\
   & \texttt{PredSub} & -0.125 &   2.0 & 259.01 & 0.52 & 0.00 & 1.14 &   3.5 & 262.19 & 0.36 & 0.00 & 1.14 \\
   &         &  0.125 &   7.2 &  72.68 & 0.44 & 0.00 & 0.97 &  11.9 &  77.99 & 0.31 & 0.00 & 0.97 \\
   &         &  0.375 &  27.1 &  19.29 & 0.41 & 0.00 & 0.90 &  45.5 &  20.44 & 0.29 & 0.00 & 0.90 \\
\bottomrule
\end{tabular}
\caption{Time (in secs) and error comparison between ASE and \texttt{PredSub} with
  $m = \lceil(\log n)^{4+b}\rceil$ for $n=150000$ and various $(d,\rho_n)$
  values under GRDPG setting.}
\label{tab3}
\end{table}

From Table~\ref{tab1}, we observe that for $n = 60000$, $d = 5$, and $\rho_n = 0.01$, the full sample ASE method took approximately $35.2$ seconds and produced a relative error of $25\%$. In contrast, when $m = 19774\ (b = 0.125)$, \texttt{PredSub} required only $2.5$ seconds, which was $14.31$ times faster than ASE, and produced a relative error of $25\%$ which was identical to ASE. Even if we consider a much smaller subsample size, i.e., $m = 5962\ (b = -0.375)$, \texttt{PredSub} took only $0.3$ seconds to run, which was $134.1$ times faster than ASE, and produced a relative error of $35\%$, which was not drastically  higher than ASE. 
In general, as $m$ was increased, the runtime for \texttt{PredSub} increased while the relative error decreased, thus providing a range of options to balance the trade-off between computational scalability and statistical accuracy. 

The advantage of \texttt{PredSub} over ASE became more substantial as $d$ was increased. Specifically, for $d=20$, \texttt{PredSub} with $m = 10858\ (b = -0.125)$ was $52$ times faster than ASE while producing similar levels of accuracy.
We observe that ASE was quite sensitive to sparsity, with the relative error decreasing as the network became denser. \texttt{PredSub} remained consistently effective in both sparse and dense regimes, providing substantial computational gains while closely mirroring the statistical accuracy of full-sample ASE for appropriate choices of $m$.
The patterns observed in Table~\ref{tab1} are repeated in Table~\ref{tab2} ($n=100000$) and Table~\ref{tab3} ($n=150000$). Note that, for computational expediency, we set $\rho_n \in \{ 0.005,0.01\}$ when $n = 150000$.

\subsection{Estimation under RDPG}
\label{sec:estimation_RDPG}
Here, we compared \texttt{PredSub} with the 
estimator proposed by \cite{chakraborty2025scalable}, which we refer to as estSS. 
For conciseness, we only considered networks with size $n = 100000$ and $d\in \{5,10\}$. As estSS is specifically designed for the RDPG setting, we modified the generative process slightly to ensure that the $d \times d$ matrix $B$ was positive definite. 
We kept the subsample size $m$ in \texttt{PredSub} as $m = \lceil(\log n)^{4+b}\rceil$ where $b \in (-0.25,0.25)$. For estSS the number of subsamples was fixed at $k=20$ and we varied the overlap size $\in \{13240,28800\}$. These choices ensure that the sizes of the subsamples handled by \texttt{PredSub} and estSS were similar, thereby providing a fair basis for comparison. 

\begin{table}[ht]
\centering
\small
\setlength{\tabcolsep}{3pt}
\sisetup{table-format=3.2}
\begin{tabular}{
  c
  l
  @{\hspace{1.5em}}S[table-format=1.2]
  S[table-format=3.1]
  S[table-format=3.1]
  S[table-format=1.2]
  S[table-format=1.2]
  S[table-format=1.1]
  S[table-format=3.1]
  S[table-format=3.1]
  S[table-format=1.2]
  S[table-format=1.2]
  S[table-format=1.1]
}
\toprule
& & &
  \multicolumn{5}{c}{$\rho_n = 0.01$} &
  \multicolumn{5}{c}{$\rho_n = 0.04$} \\
\cmidrule(lr){4-8} \cmidrule(lr){9-13}
& & &
  \multicolumn{2}{c}{Time (secs)} &
  \multicolumn{3}{c}{Accuracy} &
  \multicolumn{2}{c}{Time (secs)} &
  \multicolumn{3}{c}{Accuracy} \\
\cmidrule(lr){4-5} \cmidrule(lr){6-8} \cmidrule(lr){9-10} \cmidrule(lr){11-13}
$d$ & Method & {$b$} &
  {Mean} & {$\downarrow$ times} & {Error} & {SD} & {$\uparrow$ times} &
  {Mean} & {$\downarrow$ times} & {Error} & {SD} & {$\uparrow$ times} \\
\midrule
5  & ASE                                                                       & {--}   &  6.1 &  {--} & 0.53 & 0.00 & {--} &  19.0 &  {--} & 0.27 & 0.00 & {--} \\
   & \texttt{PredSub}                                                                   & -0.25  &  0.1 &  68.7 & 0.95 & 0.01 & 1.8  &   0.2 & 109.0 & 0.56 & 0.00 & 2.1  \\
   &                                                                           &  0.00  &  0.2 &  31.5 & 0.81 & 0.00 & 1.5  &   0.4 &  49.9 & 0.42 & 0.00 & 1.6  \\
   &                                                                           &  0.25  &  0.5 &  12.1 & 0.69 & 0.00 & 1.3  &   1.3 &  14.6 & 0.34 & 0.00 & 1.3  \\
   & \multirow{2}{*}{\begin{tabular}[c]{@{}c@{}}estSS\\ $(k=20)$\end{tabular}} &  0.00  &  0.7 &   9.0 & 1.51 & 0.05 & 2.9  &   1.0 &  19.5 & 0.64 & 0.01 & 2.4  \\
   &                                                                           &  0.25  &  1.2 &   5.0 & 1.05 & 0.02 & 2.0  &   2.8 &   6.9 & 0.46 & 0.00 & 1.7  \\
\midrule
10 & ASE                                                                       & {--}   & 22.9 &  {--} & 0.78 & 0.00 & {--} &  53.8 &  {--} & 0.35 & 0.00 & {--} \\
   & \texttt{PredSub}                                                                   & -0.25  &  0.2 & 130.7 & 1.09 & 0.00 & 1.4  &   0.3 & 160.5 & 0.66 & 0.00 & 1.9  \\
   &                                                                           &  0.00  &  0.5 &  48.4 & 0.95 & 0.00 & 1.2  &   1.0 &  55.8 & 0.56 & 0.00 & 1.6  \\
   &                                                                           &  0.25  &  1.5 &  15.2 & 0.87 & 0.00 & 1.1  &   4.0 &  13.4 & 0.47 & 0.00 & 1.3  \\
   & \multirow{2}{*}{\begin{tabular}[c]{@{}c@{}}estSS\\ $(k=20)$\end{tabular}} &  0.00  &  1.2 &  19.1 & 2.03 & 0.04 & 2.6  &   2.0 &  26.6 & 0.95 & 0.01 & 2.7  \\
   &                                                                           &  0.25  &  3.2 &   7.1 & 1.47 & 0.02 & 1.9  &   8.1 &   6.6 & 0.68 & 0.00 & 1.9  \\
\bottomrule
\end{tabular}
\caption{Time (in secs) and error comparison between ASE, \texttt{PredSub} with
  $m = \lceil(\log n)^{4+b}\rceil$ and estSS with $s \approx m$
  with $k = 20$ for $n=100000$ and various $(d,\rho_n)$ values under RDPG setting.}
\label{tab4}
\end{table}

From Table~\ref{tab4}, we observe that for $n = 100000$, $d = 5$, and $\rho_n = 0.01$, the full sample ASE method took approximately $6.1$ seconds and produced a relative error of $53\%$. Once again, as $m$ was increased, the runtime for \texttt{PredSub} increased while the relative error decreased. For instance, when $m = 32363\ (b = 0.25)$, \texttt{PredSub} required only $0.5$ seconds, which was $12$ times faster than ASE, and produced a relative error of $69\%$, approximately $1.3$ times that of ASE. For comparison, the estSS method with a similar subsample size ($s = 32360$) took $1.2$ seconds, which was $5$ times faster than ASE, and produced a relative error of $105\%$, over $2$ times larger than the ASE error. These results demonstrate that, compared to  estSS, \texttt{PredSub} provides more substantial computational savings while producing statistical accuracy that is much closer to the full sample ASE.

\citet{chakraborty2025scalable} noted that the estimation accuracy of estSS can be improved by reducing  the number of subsamples, thereby increasing the subsample size. 
However, as the computational cost grows with the subsample size, this adjustment will invariably undermine the intended scalabity of estSS. 

\subsection{Two-sample testing}
For the two-sample testing experiments, we generated $P^{(1)}$ following the procedure described in Section~\ref{sec4.1}, namely $P^{(1)} = \Pi B \Pi^{\top}$, and set $P^{(2)} = \Pi (B + \epsilon J) \Pi^{\top}$, where $J$ is the $d \times d$ matrix of all ones.
Note that $\epsilon = 0$ corresponds to the null hypothesis and $\epsilon > 0$ corresponds to the alternative hypothesis, with larger values of $\epsilon$ indicating stronger alternatives.
We then sampled adjacency matrices $A^{(1)} \sim \text{Bernoulli}(P^{(1)})$, $A^{(2)} \sim \text{Bernoulli}(P^{(2)})$, and 
tested 
$
H_0 : P^{(1)} = P^{(2)} \quad \text{versus} \quad H_1 : P^{(1)} \ne P^{(2)}$.
To evaluate the performance of the testing procedure, we repeated the above steps $r$ times, where for each replication, we generated independent network pairs $(A^{(1)}, A^{(2)})$, and applied the test to obtain an acceptance or rejection decision. 

We compared \texttt{PredSubTest} and \texttt{PureSubTest}
against the method from \cite{bhadra2025bootstrap}, referred to as ASE.
We did not evaluate the omnibus embedding method of \cite{levin2017central} as it is much slower than ASE-based methods 
(see Figure~3 of \cite{bhadra2025bootstrap} for details) and is thus too computationally expensive for the large networks used in our simulation study.
We also did not include the scalable two-sample testing method of \cite{chakraborty2025scalable}  as it is only applicable to the RDPG setting. 

We used $n \in \{20000,30000,50000\}$, $d = 10$, and $\rho_n = 0.01$.
For each method and each combination of $n, d$ and $\rho$, we computed the rejection rate from $r$ replicates, which was interpreted as the empirical level under the null hypothesis ($\epsilon = 0$) and the empirical power under the alternative hypothesis ($\epsilon > 0$). We used $B = 100$ bootstrap resamples and $r = 200$ Monte Carlo replications for all methods to obtain reliable estimates of both level and power. We also report the runtime associated with each method. The results are summarized in Table~\ref{tab:testing}.

\begin{table}[ht]
\centering
\small

\sisetup{table-format=3.2}
\begin{tabular}{
  l
  @{\hspace{1.5em}}S[table-format=1.3]
  S[table-format=1.2]
  S[table-format=1.2]
  S[table-format=1.2]
  S[table-format=1.2]
  S[table-format=1.2]
  S[table-format=3.2]
  S[table-format=3.1]
}
\toprule
& &
  \multicolumn{7}{c}{$n = 20000$} \\
\cmidrule(lr){3-9}
& &
  \multicolumn{1}{c}{Level} &
  \multicolumn{4}{c}{Power} &
  \multicolumn{2}{c}{Time (mins)} \\
\cmidrule(lr){3-3} \cmidrule(lr){4-7} \cmidrule(lr){8-9}
Method & {$b$} &
  {$\epsilon = 0$} & {$\epsilon = 0.02$} & {$\epsilon = 0.04$} &
  {$\epsilon = 0.08$} & {$\epsilon = 0.16$} &
  {Mean} & {$\downarrow$ times} \\
\midrule
ASE     & {--}   & 0.00 & 0.00 & 0.00 & 1.00 & 1.00 &  29.15 & {--} \\
\texttt{PredSubTest} & -0.250 & 0.06 & 0.29 & 0.96 & 1.00 & 1.00 &   4.81 &  6.1 \\
        & -0.125 & 0.00 & 0.03 & 0.84 & 1.00 & 1.00 &   6.73 &  4.3 \\
        &  0.000 & 0.00 & 0.00 & 0.48 & 1.00 & 1.00 &   9.92 &  2.9 \\
\texttt{PureSubTest}& -0.125 & 0.00 & 0.00 & 0.00 & 0.00 & 0.00 &   3.13 &  9.3 \\
        &  0.000 & 0.00 & 0.00 & 0.00 & 0.00 & 1.00 &   5.93 &  4.9 \\
        &  0.125 & 0.00 & 0.00 & 0.00 & 0.00 & 1.00 &  11.44 &  2.5 \\
\midrule
\multicolumn{2}{l}{} & \multicolumn{7}{c}{$n = 30000$} \\
\midrule
ASE     & {--}   & 0.00 & 0.00 & 1.00 & 1.00 & 1.00 &  85.76 & {--} \\
\texttt{PredSubTest} & -0.125 & 0.08 & 0.70 & 1.00 & 1.00 & 1.00 &  12.96 &  6.6 \\
        &  0.000 & 0.03 & 0.35 & 1.00 & 1.00 & 1.00 &  19.33 &  4.4 \\
        &  0.125 & 0.00 & 0.04 & 1.00 & 1.00 & 1.00 &  28.58 &  3.0 \\
\texttt{PureSubTest}&  0.000 & 0.00 & 0.00 & 0.00 & 0.00 & 1.00 &   9.80 &  8.7 \\
        &  0.125 & 0.00 & 0.00 & 0.00 & 0.96 & 1.00 &  18.21 &  4.7 \\
        &  0.250 & 0.00 & 0.00 & 0.00 & 1.00 & 1.00 &  33.82 &  2.5 \\
\midrule
\multicolumn{2}{l}{} & \multicolumn{7}{c}{$n = 50000$} \\
\midrule
ASE     & {--}   & 0.00 & 0.00 & 1.00 & 1.00 & 1.00 & 257.80 & {--} \\
\texttt{PredSubTest} &  0.000 & 0.10 & 0.99 & 1.00 & 1.00 & 1.00 &  40.53 &  6.4 \\
        &  0.125 & 0.03 & 0.99 & 1.00 & 1.00 & 1.00 &  55.99 &  4.6 \\
        &  0.250 & 0.00 & 0.97 & 1.00 & 1.00 & 1.00 &  82.21 &  3.1 \\
\texttt{PureSubTest}&  0.125 & 0.00 & 0.00 & 0.00 & 1.00 & 1.00 &  28.47 &  9.1 \\
        &  0.250 & 0.00 & 0.00 & 0.00 & 1.00 & 1.00 &  52.84 &  4.9 \\
        &  0.375 & 0.00 & 0.00 & 1.00 & 1.00 & 1.00 & 105.63 &  2.4 \\
\bottomrule
\end{tabular}
\caption{Level and power analysis for hypothesis testing using ASE, and \texttt{PredSubTest} and \texttt{PureSubTest}
  with $m = \lceil(\log n)^{4+b}\rceil$ and various $(d,n)$ values under GRDPG setting.}
\label{tab:testing}
\end{table}

From Table~\ref{tab:testing}, it is evident that ASE attained near‐perfect power as soon as the perturbation parameter $\epsilon$ departed from zero. However, the computational cost of ASE grew quickly with $n$, requiring approximately $29$ minutes for $n=20000$, $85$ minutes for $n=30000$, and exceeding $250$ minutes for $n=50000$, which renders it impractical for large‐scale applications. In comparison, \texttt{PredSubTest} achieved statistical performance close to ASE in terms of both size and power. Specifically, the empirical type I error rates were close to the nominal significance level, and the test attained full power for relatively small perturbations (e.g., $\epsilon \geq 0.08$ for $n = 20000$ or $\epsilon \geq 0.04$ for larger networks). \texttt{PredSubTest} was $3-7$ faster than ASE and the computational gains became increasingly substantial as $n$ grew. 

\texttt{PureSubTest}, on the other hand, was considerably more conservative under the null and required moderately larger perturbations to achieve full power, compared to \texttt{PredSubTest}. 
Although \texttt{PureSubTest} was generally faster than \texttt{PredSubTest} for the same $n$ and $m$, its lower accuracy demonstrates the importance of the out‐of‐sample estimation step in \texttt{PredSubTest}. In summary, given sufficient computational budget, \texttt{PredSubTest} appears to be the more effective method for balancing statistical reliability and computational scalability.

\section{Real Data Examples}
\label{sec6}
We applied
our testing procedures to two large real-world datasets.


\begin{enumerate}
    \item Coauth-DBLP: This dataset is available at \url{https://www.cs.cornell.edu/~arb/data/coauth-DBLP/} and consists of publications recorded in the DBLP database, where nodes represent authors and edges imply co-authorship \citep{Benson-2018-simplicial}.
We constructed binary, undirected networks for two distinct intervals: 2011–2014 and 2015–2018. Retaining only authors with at least five coauthors in each period yields networks with $n=104{,}488$ nodes and $1{,}343{,}928$ and $1{,}341{,}888$ edges, respectively. We tested whether the underlying probability matrix remains the same over time, i.e., $H_0: P_{2011-14} = P_{2015-18}$ vs. $H_1: P_{2011-14} \ne P_{2015-18}$, where $P_I$ denotes the probability matrix for period $I$.

\item Cannes 2013: This Twitter dataset \citep{omodei2015characterizing}, available at 
\url{https://manliodedomenico.com/data.php}, records retweet and mention interactions among $438{,}537$ users during the 2013 Cannes Film Festival. 
After converting to undirected binary graphs and removing isolated nodes, the network had $n=119{,}708$ nodes and $484{,}874$ edges (retweet) and $479{,}148$ edges (mention). We tested whether the underlying probabilities for retweets and mentions are identical, i.e., $H_0: P_{\text{retweet}} = P_{\text{mention}}$ vs. $H_1: P_{\text{retweet}} \ne P_{\text{mention}}$.
\end{enumerate}

For both datasets, we set $m = \lceil(\log n)^{4+b}\rceil$, where $b \in \{-0.25, 0, 0.25\}$ for \texttt{PredSubTest} and $b \in \{-0.125, 0.125, 0.375\}$ for \texttt{PureSubTest}.
We only used the RDPG setting (i.e., $q=0$) to provide a fair comparison with estSS. 

Tables~\ref{tab:real_coauth} and \ref{tab:real_cannes} show that in all settings, \texttt{PredSubTest} consistently rejected $H_0$ with a $p$-value of zero, the same as the full-sample ASE method, while being approximately $3$-$21$ times faster. \texttt{PureSubTest} produced the same inference for larger values of $b$, with even greater speedups of $8$-$215$ times. Thus, for both datasets, our methods offer substantial computational advantages over the full sample ASE approach while providing similar statistical output.
Notably, these gains are consistent across all choices of $d$, suggesting that an appropriate value of $d$ can be estimated separately prior to testing without affecting the conclusions.
The estSS method of \cite{chakraborty2025scalable} also rejects $H_0$ at comparable subsample sizes but incurs runtime similar to full-sample ASE. 

\begin{table}[ht]
\centering
\fontsize{9}{10.5}\selectfont
\setlength{\tabcolsep}{3pt}
\sisetup{table-format=3.2}
\begin{tabular}{
  l
  @{\hspace{1.5em}}S[table-format=1.3]
  S[table-format=1.3]
  S[table-format=4.2]
  S[table-format=3.1]
  S[table-format=1.3]
  S[table-format=4.2]
  S[table-format=3.1]
  S[table-format=1.3]
  S[table-format=4.2]
  S[table-format=3.1]
}
\toprule
& &
  \multicolumn{3}{c}{$d = 5$} &
  \multicolumn{3}{c}{$d = 10$} &
  \multicolumn{3}{c}{$d = 20$} \\
\cmidrule(lr){3-5} \cmidrule(lr){6-8} \cmidrule(lr){9-11}
Method & {$b$} &
  {$p$-value} & {Time (mins)} & {$\downarrow$} &
  {$p$-value} & {Time (mins)} & {$\downarrow$} &
  {$p$-value} & {Time (mins)} & {$\downarrow$} \\
\midrule
ASE     & {--}   & 0     & 1360.27 & {--}  & 0     & 1543.48 & {--}  & 0     & 1792.39 & {--}  \\
\texttt{PredSubTest} & -0.25  & 0     &   63.54 & 21.4  & 0     &   95.53 & 16.2  & 0     &  172.63 & 10.4  \\
        &  0.00  & 0     &  111.84 & 12.2  & 0     &  164.15 &  9.4  & 0     &  298.48 &  6.0  \\
        &  0.25  & 0     &  214.58 &  6.3  & 0     &  283.73 &  5.4  & 0     &  502.44 &  3.6  \\
\texttt{PureSubTest}& -0.125 & 0     &    7.54 & 180.4 & 0     &   10.01 & 154.1 & 0     &   15.84 & 113.2 \\
        &  0.125 & 0     &   32.89 &  41.4 & 0     &   40.67 &  37.9 & 0     &   57.30 &  31.3 \\
        &  0.375 & 0     &  149.60 &   9.1 & 0     &  168.76 &   9.1 & 0     &  218.92 &   8.2 \\
estSS   &  0.00  & 0     & 1398.98 &   1.0 & 0     & 1508.66 &   1.0 & 0.505 & 1794.72 &   1.0 \\
        &  0.25  & 0     & 1407.00 &   1.0 & 0     & 1493.93 &   1.0 & 0     & 1690.03 &   1.1 \\
\bottomrule
\end{tabular}
\caption{Hypothesis testing on the Coauth-DBLP dataset using ASE, \texttt{PredSubTest} and \texttt{PureSubTest}with
  $m = \lceil(\log n)^{4+b}\rceil$ and estSS with $s \approx m$
  with $k = 20$ and various $d$ values under RDPG setting.}
\label{tab:real_coauth}
\end{table}

\begin{table}[ht]
\centering
\fontsize{9}{10.5}\selectfont
\setlength{\tabcolsep}{3pt}
\sisetup{table-format=3.2}
\begin{tabular}{
  l
  @{\hspace{1.5em}}S[table-format=1.3]
  S[table-format=1.3]
  S[table-format=4.2]
  S[table-format=3.1]
  S[table-format=1.3]
  S[table-format=4.2]
  S[table-format=3.1]
  S[table-format=1.3]
  S[table-format=4.2]
  S[table-format=3.1]
}
\toprule
& &
  \multicolumn{3}{c}{$d = 5$} &
  \multicolumn{3}{c}{$d = 10$} &
  \multicolumn{3}{c}{$d = 20$} \\
\cmidrule(lr){3-5} \cmidrule(lr){6-8} \cmidrule(lr){9-11}
Method & {$b$} &
  {$p$-value} & {Time (mins)} & {$\downarrow$} &
  {$p$-value} & {Time (mins)} & {$\downarrow$} &
  {$p$-value} & {Time (mins)} & {$\downarrow$} \\
\midrule
ASE     & {--}   & 0     & 1686.39 & {--}  & 0     & 1924.71 & {--}  & 0     & 2273.25 & {--}  \\
\texttt{PredSubTest} & -0.25  & 0     &   77.14 & 21.9  & 0     &  115.82 & 16.6  & 0     &  201.02 & 11.3  \\
        &  0.00  & 0     &  135.48 & 12.4  & 0     &  201.91 &  9.5  & 0     &  354.78 &  6.4  \\
        &  0.25  & 0     &  252.64 &  6.7  & 0     &  347.48 &  5.5  & 0     &  597.62 &  3.8  \\
\texttt{PureSubTest}& -0.125 & 0     &    7.84 & 215.2 & 0     &   10.55 & 182.4 & 0     &   16.31 & 139.4 \\
        &  0.125 & 0     &   36.63 &  46.0 & 0     &   45.57 &  42.2 & 0     &   62.89 &  36.1 \\
        &  0.375 & 0     &  226.85 &   7.4 & 0     &  264.05 &   7.3 & 0     &  265.16 &   8.6 \\
estSS   &  0.00  & 0     & 1765.17 &   1.0 & 0     & 1930.51 &   1.0 & 0     & 2224.28 &   1.0 \\
        &  0.25  & 0     & 1746.38 &   1.0 & 0     & 1873.81 &   1.0 & 0     & 2059.37 &   1.1 \\
\bottomrule
\end{tabular}
\caption{Hypothesis testing on the Cannes 2013 dataset using ASE, \texttt{PredSubTest} and \texttt{PureSubTest}with
  $m = \lceil(\log n)^{4+b}\rceil$ and estSS with $s \approx m$
  with $k = 20$ and various $d$ values under RDPG setting.}
\label{tab:real_cannes}
\end{table}

\section{Discussion}\label{sec7}

In this article, we proposed subsampling-based methods for estimation and two-sample hypothesis testing in networks. The methodology was developed under the GRDPG framework, with theoretical guarantees of consistency and numerical experiments showing substantial runtime improvements. Even when efficient $\order(n^2)$ spectral decomposition algorithms are used for the full-sample methods, our approach yields significant time savings, and its compatibility with parallel computing further enhances scalability. For this work, we only considered uniform random sampling to choose the subsample, but one can always use other sampling techniques as required. We also conducted some simulations using sparsity-based sampling, but only obtained similar accuracy with higher computation time and hence did not report these results here. 

Our primary focus was testing network equivalence under the null hypothesis $H_0: P^{(1)} = P^{(2)}$, but the framework can be extended to more general settings, such as scaled hypotheses $H_0: P^{(1)} = cP^{(2)}$ for $c>0$, or comparisons of network statistics (e.g., degree distributions or triangle counts). 
More broadly, predictive subsampling can be adapted to other inferential tasks within the GRDPG as well as other more general network models, though such extensions will require additional theoretical and computational advances. These directions offer promising research problems for developing statistically rigorous and computationally scalable methods for statistical inference with massive networks.







\bibliographystyle{plainnat}
\bibliography{sample}
\clearpage
\appendix
\section*{Appendix}
\label{appA}
\setcounter{theorem}{0}
\numberwithin{equation}{section}

The Appendix contains the technical proofs of the theoretical results. Codes can be found at the following {GitHub repository}: \url{https://github.com/ArpanK72/PredSub/tree/main}.

We begin by stating a collection of lemmas and known results, along with proofs as appropriate, that will be used throughout the subsequent analysis.

 \begin{result}[Weyl's Inequality, \cite{weyl1912asymptotic}]\label{weyl}
     Let $A$ and $B$ be symmetric $n \times n$ matrices. Then
     $$\max_i |\lambda_i(A) - \lambda_i(B)| \leq \|A-B\|.$$
 \end{result}
 
\begin{result}[Singular Value Perturbation, \cite{horn2012matrix}] \label{singpert}
    Let $A$ and $B$ be $n \times m$ matrices and $q=\min\{m,n\}$. Let
$\sigma_1(A)\ge \cdots \ge \sigma_q(A)$ and
$\sigma_1(B)\ge \cdots \ge \sigma_q(B)$ be the
singular values of $A$ and $B$, respectively. Then
$$
\max_{i} \left|\sigma_i(A)-\sigma_i(B)\right| \le \|A-B\|.
$$
\end{result}

\begin{result}[Matrix Bernstein Inequality, \cite{tropp2012user}]
    Consider a finite sequence $\{X_k\}$ of independent, random, self-adjoint matrices with dimension $d$. Assume that each random matrix satisfies
$\mathbb{E}[X_k] = 0$ and $ \lambda_{\max}(X_k) \leq R $ almost surely. Then, for all $t \geq 0$,
\begin{align*}
\prob \left( \lambda_{\max} \left( \sum_k X_k \right) \geq t \right) \leq d \cdot \exp \left( -\frac{t^2}{2 \sigma^2 + \frac{2R t}{3}} \right),   \quad \text{where  } \sigma^2 := \left\|\sum_k\E[X_k^2]\right\|. 
\end{align*} \label{matbennstein}
\end{result}

\begin{lemma}\label{lem1}
    Let $P$ be a symmetric matrix satisfying Assumption~2-3. Then the eigen vectors of $P$ are bounded coherent i.e., there exists a constant $c > 0$ not depending on $n$ such that
    $\|U_P\|_{2\to\infty} \leq \dfrac{c}{\sqrt{n}}$
    where $P = U_PD_PU_P^\top$ and $D_P$ is the diagonal matrix with the eigenvalues of $P$ on its diagonal. Consequently $\|U_P |D_P|^{1/2}\|_{2\to\infty} \leq K \rho_n^{1/2}$
    for some constant $K > 0$ not depending on $n$.
\end{lemma}

\begin{proof}
Recall that the entries of $P$ are bounded between $c_1 \rho_n$ and $c_2 \rho_n$. Hence, by the Perron-Frobenius theorem, $nc_1 \rho_n \leq \sigma_1(P) \leq nc_2 \rho_n$. Assumption~3 then implies
    \begin{equation}
    \label{eq:lowerbd_sigma_d}
       \sigma_d(P)  \geq \dfrac{\sigma_1(P)}{c_0} 
        \geq \dfrac{nc_1 \rho_n}{c_0}.
    \end{equation}
    Next we have $P U_P = U_P D_P$ so that, for any $i \in [n]$,
    \begin{gather}
    \label{eq:upperbound_PU_P}
        \|(P U_P)_i\|^2 = \|P_i U_P\|^2 \leq \|P_i\|^2 \le nc_2^2 \rho_n^2, \\
    \label{eq:lowerbound_PU_P}
        \|(P U_P)_i\|_2^2 = \|(U_P D_P)_i\|^2 \ge \sigma_d^2(P) \|(U_{P})_{i}\|^2.
    \end{gather}
    where $(U_P)_{i}$ denote the $i$th row of $U_P$.
    Combining Eq.~\eqref{eq:lowerbd_sigma_d} through Eq.~\eqref{eq:lowerbound_PU_P} we obtain
    \begin{align*}
       \|(U_{P})_i\|  \leq \dfrac{c_0 c_2}{\sqrt{n}c_1}.
        \end{align*}
        As $i$ is arbitrary, we have $\|U_P\|_{2\to\infty} \leq \dfrac{c}{\sqrt{n}}$ where $c = c_0c_2/c_1$. Finally we have
    \begin{align*}
        \|U_{P}|D_P|^{1/2}\|_{2\to\infty} \le \|U_P\|_{2 \to \infty} \times (\sigma_1(P))^{1/2}
        \leq K \rho_n^{1/2} 
    \end{align*}
    where $K = c \sqrt{c}_2 = c_0 c_2^{3/2}/c_1$. 
\end{proof}

We next state a lemma relating the singular values of $P$ to that of its subsampled $P_S$. The lemma is a refinement of Theorem~1 in \cite{bhattacharjee2024sublinear} when 
$P$ is low-rank.

\begin{lemma} \label{lem2}
    Let $P$ be a symmetric matrix that satisfies Assumptions~2-3. 
    Let $S \subseteq [n]$ of $m$ vertices be formed by including each index independently with probability $m/n$ and let $P_S$ be the corresponding principal submatrix of $P$. 
    Denote by $\sigma_i(P)$ and $\sigma_i(P_S)$ the $i$th largest singular value of $P$ and $P_S$, respectively. Define 
    $$\Tilde{\sigma}_i(P) = \begin{cases}
        \frac{n}{m} \sigma_i(P_S) & \text{if  } i\in [d]\\
        0 & \text{otherwise}
    \end{cases} $$
    If $m \geq \dfrac{4(c+1)\log n}{\epsilon^2}$ for some constant $c>0$, then we have
    $$|\sigma_i(P) - \Tilde{\sigma}_i(P)| \leq \epsilon n \rho_n$$
    with probability at least $1 - n^{-c}$.
\end{lemma}

\begin{proof}
We follow the proof ideas in \cite{bhattacharjee2024sublinear}. 
Let $R \in \mathbb{R}^{n\times |S|}$ be the scaled sampling matrix such that $R^\top P R = \frac{n}{m} P_S$. Now $P = U_P D_P U_P^\top = U_P|D_P|^{1/2}I_{p,q}|D_P|^{1/2} U_P^\top$ and hence $\frac{n}{m}P_S$ have the same eigenvalues as $|D_P|^{1/2} U_P^\top R R^\top U_P|D_P|^{1/2}I_{p,q}$. 
    Let 
    \[E = |D_P|^{1/2} U_P^\top R R^\top U_P|D_P|^{1/2}I_{p,q} - D_P, \quad F = |D_P|^{1/2} U_P^\top R R^\top U_P|D_P|^{1/2} - |D_P|\]
    and note that $E = FI_{p,q}$. For all $i \in [n]$, let $U_{Pi}$ denote the $i$th row of $U_P$.
    Define mutually indepedent random {\em symmetric} matrices $Y_1, Y_2, \dots, Y_n$ where 
    $$ Y_i = \begin{cases}
        \frac{n}{m} |D_P|^{1/2} U_{Pi}^\top U_{Pi}|D_P|^{1/2} & \text{with probability  } m/n,\\
        0 & \text{otherwise}.
    \end{cases}$$
    Let $Z_i = Y_i - \E[Y_i]$. Then $\sum_{i=1}^n Z_i = F$. Next observe that, for any $i$,  
    \[\|Z_i\|_2 \leq \max\left\{1, \frac{n}{m}-1\right\}\| |D_P|^{1/2} U_{Pi}^\top U_{Pi}|D_P|^{1/2} \|\leq \frac{n}{m}\| U_{Pi}|D_P|^{1/2} \|^2 \leq \dfrac{nK^2 \rho_n}{m} \]
    where the last inequality follows from Lemma~\ref{lem1}. We also have
    \begin{align*}
       \sum_{i=1}^n \E[Z_i^2] & = \sum_{i=1}^n \left[ \dfrac{m}{n} \left(\dfrac{n}{m}-1\right)^2 + \left(1-\dfrac{m}{n}\right)\right] \left(|D_P|^{1/2} U_{Pi}^\top U_{Pi}|D_P| U_{Pi}^\top U_{Pi}|D_P|^{1/2}\right) \\
       & \preceq \dfrac{n}{m} \sum_{i=1}^n \| U_{Pi}|D_P|^{1/2} \|^2 \left(|D_P|^{1/2}  U_{Pi}^\top U_{Pi}|D_P|^{1/2}\right) \\
       & \preceq \dfrac{nK^2 \rho_n}{m}\sum_{i=1}^n||D_P|^{1/2}  U_{Pi}^\top U_{Pi}|D_P|^{1/2}
        = \dfrac{nK^2 \rho_n}{m}|D_P|
        \end{align*}
        where the penultimate inequality follows from Lemma~\ref{lem1} and $\preceq$ denote the Loewner ordering for positive semidefinite matrices. We thus have
        \[  \left\|\sum_{i=1}^n \E[Z_i^2]\right\| \preceq \dfrac{n K^2 \rho_n \sigma_1(P)}{m}\\
        \leq \dfrac{n^2 K^2 \rho_n c_2}{m} \]
    Applying the matrix Bernstein's inequality (see Result~\ref{matbennstein}), we obtain
    \begin{align*}
        \prob\left( \|F\| \geq \epsilon n \rho_n \right) & \leq 2d\exp\left\{ - \dfrac{\epsilon^2 n^2 \rho_n^2}{\dfrac{2n^2K^2 \rho_n^2 c_2}{m} + \dfrac{2\epsilon n^2K^2 \rho_n^2}{3m}} \right\} \leq 2d\exp\left\{ - \dfrac{3m\epsilon^2}{6K^2c_2 + 2\epsilon K^2} \right\}
    \end{align*}
    Choosing $m\geq \dfrac{4(c+1)\log n}{\epsilon^2}$ for some large enough constant $c$, we get 
    \[\prob(\|E\| \geq \epsilon n \rho_n) \leq \prob(\|F\| \geq \epsilon n \rho_n) \leq n^{-c}. \]
    Furthermore, we have
    \begin{equation*}
    \begin{split}
        |\tilde{\sigma}_i(P) - \sigma_i(P)| &= \left|\tfrac{n}{m} \sigma_i(P_S) - \sigma_i(D_P)\right| = 
        \sigma_i\left(|D_P|^{1/2} U_P^\top R R^\top U_P|D_P|^{1/2}I_{p,q}\right) - \sigma_i(D_P) 
        \leq \|E\|
    \end{split}
    \end{equation*}
    where the inequality follows from Result~\ref{singpert} and the definition of $E$. In conclusion we have
    \[ \prob\left(\left| \tilde{\sigma}_i(P) - \sigma_i(P)\right| \leq \epsilon n \rho_n \right) \geq \prob(\|E\| \leq \epsilon n \rho_n) \geq 1 - n^{-c} \]
    as claimed. 
\end{proof}

\begin{lemma}
\label{lem:ct}
    Consider $A$ to be a $n\times m$ matrix with $A_{ij} \sim \mathrm{Bernoulli}(P_{ij})$ where $P$ satisfies assumption~2. Define $E = A-P$. Then for any $m\times d$ matrix $V$ independent of $E$,
    $$\|EV\|_{2\to\infty} \leq 2\sqrt{2\rho_n v(\nu,d,n)}\|V\| + \dfrac{4}{3}\|V\|_{2\to\infty}v(\nu,d,n)$$
    with probability at least $1-n^{-(\nu -1)}$ where $v(\nu,d,n) = \log(n^\nu 5^d)$.
    \label{lem3}
\end{lemma}

\begin{proof}
Our result is a slight modification of Lemma 8 of \cite{tang2025eigenvector}, and we follow similar steps for the proof. The main difference is that, in contrast with \cite{tang2025eigenvector}, we don't need $E$ to be a square matrix. First note that,
    $$\|EV\|_{2\to\infty} = \max_{i} \|(EV)_i\| = \max_i \|e_i^\top EV\|$$
    where $e_i$ denote the $i$-th elementary basis vector in $\R^n$. Now for any $i$,
    \[\|e_i^\top EV\|  = \sup_{u\in \mathcal{S}} e_i^\top EVu, \quad \mathcal{S} = \{u \in \R^d : \|u\|=1 \}.\] Let $\mathcal{S}_{1/2}$ be a $1/2$-cover of $\mathcal{S}$ with respect to $\ell_2$-norm. Then, for any $u\in\mathcal{S}$ and $w\in\mathcal{S}_{1/2}$ with $\|u-w\|\leq 1/2$ we have,
    \begin{align*}
        e_i^\top EV(u-w) 
         \le \|e_i EV\| \times \|u-w\| 
         \leq \dfrac{1}{2}\|e_i EV\|
    \end{align*}
    and hence
    \[ \|e_i E V\| = \sup_{u\in\mathcal{S}} e_i^\top EVu  \leq \sup_{w\in\mathcal{S}_{1/2}} e_i^\top EVw +  \dfrac{1}{2}\|e_i EV\| \]
    so that $\|e_i E V\| \leq 2\sup_{w\in\mathcal{S}_{1/2}} e_i^\top EVw$. 
    Now for any $w\in \mathcal{S}_{1/2}$, we have,
    $$e_i^\top EVw = \sum_{j=1}^m E_{ij}(Vw)_j$$
    where $E_{ij}(Vw)_j$ are mutually independent random variables
    with
    $\E(E_{ij}(Vw)_j) = 0$ and $|E_{ij}(Vw)_j| \leq \|Vw\|_\infty \leq \|V\|_{2\to\infty}$ almost surely. 
    Furthermore, 
    \[ \sum_{j=1}^m \E(E_{ij}^2(Vw)_j^2) = \sum_{j=1}^m (Vw)_j^2 \var[A_{ij}] = \sum_{j=1}^m P_{ij}(1-P_{ij}) (Vw)_j^2 \leq \rho_n \|Vw\|^2 \leq \rho_n \|V\|^2.\]
    Applying Bernstein's inequality we obtain
    \begin{gather*}
        \prob\left( e_i^\top EVw > t \right)\leq \exp\left\{ \dfrac{-t^2}{2\rho_n \|V\|^2+ 2\|V\|_{2 \to \infty} t/3} \right\}
    \end{gather*}
    Let $t_1 = \dfrac{2}{3}v(\nu,d,n)\|V\|_{2\to\infty} + \sqrt{2\rho_nv(\nu,d,n)}\|V\|$. Then
    \begin{align*}
        \dfrac{t_1^2}{2\rho_n\|V\|^2 + 2\|V\|_{2 \infty} t_1/3} & \geq \dfrac{\left( \frac{2}{3}v(\nu,d,n)\|V\|_{2\to\infty} + \sqrt{2\rho_nv(\nu,d,n)}\|V\| \right)^2}{2\rho_n\|V\| + \frac{2}{3}\|V\|_{2\to\infty} \left( \frac{2}{3}v(\nu,d,n)\|V\|_{2\to\infty} + \sqrt{2\rho_nv(\nu,d,n)}\|V\| \right)} \\
        & \geq v(\nu,d,n).
    \end{align*}
    and hence 
    \[ \prob\left( e_i^\top EVw > t_1 \right)  \leq \exp\{-v(\nu,d,n)\} = n^{-\nu}5^{-d}. \]
    The volumetric argument in Corollary 4.2.13 in \cite{vershynin2018high} yields
    $|\mathcal{S}_{1/2}|\leq 5^d$. Thus, taking a union bound over all $w\in\mathcal{S}_{1/2}$, we have
    \begin{align*}
    \prob\left( \| e_i^\top EV\| > 2t_1 \right) \leq 
        \prob\left( \sup_{w\in\mathcal{S}_{1/2}} e_i^\top EVw > t_1 \right) &  \leq n^{-\nu}. 
    \end{align*}
    Taking another union bound over all $i\in [n]$ we obtain
    $$\prob\left( \|EV\|_{2\to\infty} > 2t_1 \right) \leq n^{-(\nu-1)}.$$
    In summary we have
    $$\|EV\|_{2\to\infty} \leq 2\sqrt{2\rho_n v(\nu,d,n)}\|V\| + \dfrac{4}{3}\|V\|_{2\to\infty}v(\nu,d,n)$$
    with probability at least $1-n^{-(\nu -1)}$.
\end{proof}

\begin{lemma}\label{lem:subsample}
    Assume the setting in Theorem~1. Let $A_S,P_S$ be the principal submatrix of $A,P$ respectively corresponding to the subset $S$ of size $m$ chosen uniformly at random from $[n]$. Denote by $\hat{X}_S$ the adjacency spectral embedding of $A_S$.
    Given assumption~1-4 holds, there exists an orthogonal transformation $W \in \mathbb{O}(d)$ such that
     $$ \|\hat{X}_SW - X_S \|_{2\to\infty} = \order_p\left( \sqrt{\dfrac{\log m}{m}}\right).$$
\end{lemma}

\begin{proof}
    Clearly, $A_S\thicksim \text{GRDPG}_{p,q}(P_S)$. Now we need to verify that Assumptions 1-5 in \cite{xie2024entrywise} hold for $A_S$ and $P_S$. First note that, by the Perron-Frobenius theorem, we have $\sigma_1(P_S) \leq m \rho_n c_2$ almost surely. 
    Next recall Lemma~\ref{lem2}. Then
    \begin{align*}
      \prob\left(\left| \sigma_r(P_S) - \dfrac{m}{n}\sigma_r(P)\right| \leq \epsilon m \rho_n\right) \geq 1 - n^{-c}.\\
    \end{align*}
    Furthermore, by Eq.~\eqref{eq:lowerbd_sigma_d}, $\sigma_d(P) \geq n c_1 \rho_n/c_0$ almost surely 
    and hence, with probability at least $1 - n^{-c}$, 
    \begin{equation}
    \label{eq:lowerbd_sigma_Ps}
        \sigma_d(P_S) \geq \dfrac{m \rho_n c_1}{c_0} - \epsilon m \rho_n.
    \end{equation}
    Now suppose the condition in Eq.~\eqref{eq:lowerbd_sigma_Ps} holds. 
    Let $\epsilon = \log^{-\alpha/2}{n}$ for some $\alpha>0$. Then, for $n$ sufficiently large, there exists positive constants 
    $k < c_1/c_0 - \epsilon$ and $k' > c_2/c_0 + \epsilon$ such that  
    \begin{align*}
        m \rho_n k \leq \sigma_d(P_S) \leq m \rho_n k'.
    \end{align*}
    Next we have, for any $i \in [n]$,
    $$mc_2^2 \rho_n^2 \geq \|(P_SU_{P_S})_i\|^2 = \|(U_{P_S} D_{P_S})_i\|^2 \geq \sigma_d^2(P_S)\|U_{P_S i}\|^2 \geq m^2k^2 \rho_n^2 \|U_{P_S i}\|^2 $$
    which implies $\|U_{P_S i}\|^2 \leq \dfrac{c_2^2}{mk^2}$. As $i$ is arbitrary, we have $\|U_{P_S}\|_{2\to\infty} \leq \dfrac{c_2}{k\sqrt{m}}$. Hence
    $$\|X_S\|_{2\to\infty}\leq \|U_{P_S}|D_{P_S}|^{1/2}\|_{2\to\infty} \leq \|U_{P_S}\|_{2\to\infty}\sqrt{\sigma_1(P_S)}\leq \dfrac{c_2}{k\sqrt{m}} \times \sqrt{m c_2 \rho_n} \leq \dfrac{c_2^{3/2} \rho_n^{1/2}}{k}.$$
Thus, Assumption~1 in \cite{xie2024entrywise} holds. Assumption 2 in \cite{xie2024entrywise} is exactly the same as our Assumption~4. Assumption~3 in \cite{xie2024entrywise} also holds for $E_S = A_S - P_S$ because one can set $[E_{S2}]_{ij} = 0$ for all $i,j$.

Assumption~4 and 5 require a bit of work but is mutatis mutandis as in \cite{xie2024entrywise} and we omit the details. In particular to verify Assumption~5 in \cite{xie2024entrywise} we can follow the same arguments but replace $\Delta$ there by $\Delta_m = (m\rho_n)^{-1} X_S^{\top} X_S$ and then bound 
\begin{equation}
    \|A_S - P_S\| \leq K \sqrt{m \rho_n} \label{eq:lei_rinaldo}
\end{equation}
with high probability using \cite{lei2015consistency}.

Next, define 
\begin{equation}
    W = \mathrm{diag}(W_+,W_-)\label{W}
\end{equation}
where $W_\pm = W_{2\pm}W_{1\pm}^\top$ with $U_{P\pm}^\top U_{A\pm} = W_{1\pm}\mathrm{diag}\{\sigma_1(U_{P\pm}^\top U_{A\pm}),\cdots,\sigma_d(U_{P\pm}^\top U_{A\pm})\}W_{2\pm}^\top$. Here $W_{1+},W_{2+}\in \mathbb{O}(p)$ and $W_{1-},W_{2-}\in \mathbb{O}(q)$. Then, by Theorem~3.2 in \cite{xie2024entrywise}, we have,
\begin{equation}
    \|\hat{X}_{S\pm}W_\pm - (\pm) A_S X_{S\pm}(X_{S\pm}^\top X_{S\pm})^{-1}\|_{2\to\infty} =\order_p\left(  \dfrac{ \|X_S\|_{2\to\infty}\|U_{P_S}\|_{2\to\infty}\log m}{\sqrt{m}\rho_n \lambda_d(\Delta_m)} \right)\label{Xie3.2}
\end{equation}
Now we can write,
\begin{align*}
    \hat{X}_SW - X_S 
    & = \hat{X}_SW - P_S X_S(X_S^\top X_S)^{-1}I_{p,q} \\
    & = \hat{X}_SW - A_S X_S(X_S^\top X_S)^{-1}I_{p,q} + (A_S-P_S) X_S(X_S^\top X_S)^{-1}I_{p,q}.
    \end{align*}
Let $M^{\dagger} = M(M^{\top} M)^{-1}$ denote the Moore-Penrose pseudoinverse of $M^{\top}$ when $M^{\top}$ has full-row rank. Then
\begin{align*}
\|\hat{X}_SW - X_S\|_{2\to\infty} & \leq \|\hat{X}_SW - A_S X_S^{\dagger} I_{p,q}\|_{2\to\infty} + \|(A_S-P_S) X_S^{\dagger} I_{p,q}\|_{2\to\infty}
\end{align*}
Now, $(A_S-P_S) X_S^{\dagger} I_{p,q} = (A_S-P_S) U_{P_S} |D_{P_S}|^{-1/2}I_{p,q}$ and
\begin{align*}
    \|(A_S-P_S) X_S^{\dagger}I_{p,q}\|_{2\to\infty} & = \|(A_S-P_S) U_{P_S} |D_{P_S}|^{-1/2}I_{p,q}\|_{2\to\infty}\\
    & \leq \| (A_S-P_S) U_{P_S} \|_{2\to\infty} \times \||D_{P_S}|^{-1/2}\| \\
    & \leq \| (A_S-P_S) U_{P_S} \|_{2\to\infty} \times \sqrt{\dfrac{1}{\sigma_d(P_S)}} \\
    & = \order_p\left(\sqrt{\rho_n \log m} \times \dfrac{1}{\sqrt{m \rho_n k}}\right) \\
    & = \order_p\left( \sqrt{\dfrac{\log m}{m}}\right)
\end{align*}
where the penultimate equality follows from Lemma~\ref{lem3}. Furthermore, by Eq.~\eqref{Xie3.2},
\begin{align*}
    \|\hat{X}_SW - A_S X_S^{\dagger} I_{p,q}\|_{2\to\infty} & \leq \|\hat{X}_{S+}W_+ - A_S X_{S+}^{\dagger}\|_{2\to\infty} + \|\hat{X}_{S-}W_- - A_S X_{S-}^{\dagger} \|_{2\to\infty} \\
    & =\order_p\left( \dfrac{\chi \|X_S\|_{2\to\infty}\|U_{P_S}\|_{2\to\infty}\log m}{\sqrt{m}\rho_n \lambda_d(\Delta_m)}\right) \\
    & \leq  \order_p\left(\dfrac{\log m}{m\sqrt{\rho_n}}\right).
\end{align*}
Combining the above bounds we obtain
\begin{align*}
    \|\hat{X}_SW - X_S\|_{2\to\infty}  & = \order_p\left(\dfrac{\log m}{m\sqrt{\rho_n}}\right) + \order_p\left( \sqrt{\dfrac{\log m}{m}}\right) = \order_p\left( \sqrt{\dfrac{\log m}{m}}\right).
\end{align*}
\end{proof}

\subsection{Proof of Theorem 1}

Now Lemma~\ref{lem:subsample} gives us for any uniformly chosen random subsample $S$,
$$\|\hat{X}_SW - X_S \|_{2\to\infty} = \order_p\left( \sqrt{\dfrac{\log m}{m}}\right)$$
where $W$ is defined as in Eq.~\eqref{W}.

Next using the same $W$ and by the definition of $\hat{X}_{S^{c}}$ (see Algorithm~1) we have
\begin{align*}
    \hat{X}_{S^c} - P_{S^c,S}U_{P_S}|D_{P_S}|^{-1/2}I_{p,q}W^\top & = A_{S^c,S}U_{A_S}|D_{A_S}|^{-1/2}I_{p,q}- P_{S^c,S}U_{P_S}|D_{P_S}|^{-1/2}I_{p,q}W^\top \\
& = \left[R_+ \mid -R_- \right],
\end{align*}
where
\begin{align*}
  R_+ &= A_{S^c,S}U_{A_S+}|D_{A_S+}|^{-1/2} - P_{S^c,S}U_{P_S+}|D_{P_S+}|^{-1/2}W_+^\top, \\
R_- &= A_{S^c,S}U_{A_S-}|D_{A_S-}|^{-1/2} - P_{S^c,S}U_{P_S-}|D_{P_S-}|^{-1/2}W_-^\top. 
\end{align*}
We can further decompose $R_+$ and $R_{-}$ as
\begin{align*}
    R_{\pm} 
   &  = E_{S^c,S} U_{A_S\pm} |D_{A_S\pm}|^{-1/2} 
    + P_{S^c,S} (U_{A_S\pm} - U_{P_S\pm}W_+^\top) |D_{A_S\pm}|^{-1/2} +  P_{S^c,S} U_{P_S+} Q_{2\pm}
\end{align*}
where $E_{S^c,S} = A_{S^c,S} - P_{S^c,S}$ and $Q_{2\pm}= W_{\pm} |D_{A_S\pm}|^{-1/2} - |D_{P_S\pm}|^{-1/2} W_\pm^\top$. 
Next we have
\begin{align}
\label{eq:Rpm}
    \|R_\pm\|_{2\to\infty} & \leq  (\|R_{1\pm}\|_{2\to\infty} + \|R_{2\pm} \|_{2\to\infty}) \times \||D_{A_S\pm}|^{-1/2}\| + \|P_{S^c,S} U_{P_S\pm} Q_{2\pm}\|_{2\to\infty}
\end{align}
where $R_{1\pm} = E_{S^c,S} U_{A_S\pm}$ and $R_{2\pm} = P_{S^c,S} (U_{A_S\pm} - U_{P_S\pm}W_\pm^\top)$.
By Lemma~B.3 of \cite{xie2024entrywise},
\begin{equation*}
    \|Q_{2\pm}\|_F = \| W_\pm^\top |D_{A_S\pm}|^{-1/2} - |D_{P_S\pm}|^{-1/2} W_\pm^\top \|_F = \order_p\left( \dfrac{\sqrt{t}}{(m\rho_n)^{3/2}} \right)
\end{equation*}
for any $t\to \infty$, and hence, with $t = c\log m$ for some constant $c>0$, we obtain
\begin{equation}
    \|Q_{2\pm}\|_F = \order_p\left( \dfrac{\sqrt{\log m}}{(m\rho_n)^{3/2}} \right).
\end{equation}
Next we have
\begin{equation}
   \|P_{S^c,S}\|_{2\to\infty}  = \Theta(\sqrt{m}\rho_n)  \quad \text{and}\quad 
   \||D_{P_S\pm}|^{-1/2}\|_2  \leq (\sigma_d(P_S))^{-1/2} = \order_p\left((m \rho_n)^{-1/2}\right). 
\end{equation}
The above bound for $\||D_{P_S\pm}|^{-1/2}\|$ and Weyl's inequality (see Result~\ref{weyl}) together implies
\begin{equation}
    \||D_{A_S\pm}|^{-1/2}\| \leq (\sigma_d(A_S))^{-1/2} = \order_p\left((m \rho_n)^{-1/2}\right).
\end{equation}
Note that, $U_{A_S\pm}$ and $U_{P_S\pm}$ depends only on the subsampled vertices, while $A_{S^c,S}$ along with $P_{S^c,S}$ involves only the connections which are independent of the subgraph $A_{S}$. Clearly we have $\|U_{P_S\pm}\|=\|U_{A_S\pm}\|=1$. Using Lemma B.5 of \cite{xie2024entrywise} we also have
$$\|U_{A_S\pm}\|_{2\to\infty} = \order_p\left( \|U_{P_S\pm}\|_{2\to\infty} \right)=\order_p\left( \|U_{P_S}\|_{2\to\infty} \right)= \order_p\left(m^{-1/2}\right)$$
Applying Lemma~\ref{lem3} we obtain
\begin{equation}
      \|R_{1\pm}\|_{2\to\infty}=\order_p\left(\sqrt{\rho_n \log n}\right),
\end{equation}
and hence
\begin{align*}
  \|R_{1\pm}\|_{2\to\infty} \times \| |D_{A_S\pm}|^{-1/2}\|_2 & = \order_p\left(\sqrt{\dfrac{\log n}{m}}\right) 
\end{align*}
We next bound $R_{2,\pm}$. In particular
\begin{equation}
\begin{split}
\label{eq:R2_bdd}
    \|R_{2\pm}\|_{2\to\infty} 
    & \leq \|P_{S^c,S}\|_{2\to\infty}\|U_{A_S\pm} - U_{P_S\pm}U_{P_S\pm}^\top U_{A_S\pm} +  U_{P_S\pm}(U_{P_S\pm}^\top U_{A_S\pm} - W_\pm^\top)\|\\
    & \leq \|P_{S^c,S}\|_{2\to\infty} \{ \|U_{A_S\pm} - U_{P_S\pm}U_{P_S\pm}^\top U_{A_S\pm}\|_2 + \|U_{P_S\pm}^\top U_{A_S\pm} - W_\pm^\top\|_2 \} \\
    & \leq \|P_{S^c,S}\|_{2\to\infty} \{\|\sin \Theta(U_{P_S\pm},U_{A_S\pm})\|_2 + \|\sin \Theta(U_{P_S\pm},U_{A_S\pm})\|_2^2\} \\
    & \leq \sqrt{m}\rho_n \times \left\{ \dfrac{2\|A_S - P_S\|}{m\rho_n\lambda_d(\Delta_m)} + \dfrac{4\|A_S - P_S\|^2}{m^2\rho_n^2\lambda_d^2(\Delta_m)}  \right\}\\
    & = \sqrt{m}\rho_n \times \left\{ \order_p\left(\dfrac{1}{\sqrt{m\rho_n}}\right) +\order_p\left(\dfrac{1}{m\rho_n}\right)  \right\} = \order_p(\sqrt{\rho_n})
    \end{split}
\end{equation}
where the third inequality follows from Lemma~6.7 in \cite{cape2019two} and the fourth inequality follows from the Davis-Kahan theorem \citep{davis70,samworth}.
We therefore have
$$ \|R_{2\pm}\|_{2\to\infty} \times \||D_{A_S\pm}|^{-1/2}\| = \order_p\left(m^{-1/2}\right).$$
Lastly, we have,
\begin{align*}
    \|P_{S^c,S} U_{P_S\pm} Q_{\pm}\|_{2\to\infty} & \leq \|P_{S^c,S}\|_{2\to\infty} \times \|Q_{2\pm}\| 
     = \order_p\left(\sqrt{\dfrac{m\rho_n^2 \log m}{m^3 \rho_n^3}}\right) 
     = \order_p\left(\sqrt{\dfrac{\log m}{m^2 \rho_n}}\right).
\end{align*}
Substituting the above bounds into Eq.~\eqref{eq:Rpm} we obtain
\begin{align*}
    \|R_\pm\|_{2\to\infty} &\leq \order_p\left(\sqrt{\dfrac{\log n}{m}}\right) + \order_p\left(m^{-1/2}\right)+ \order_p\left(\sqrt{\dfrac{\log m}{m^2 \rho_n}}\right)  
     = \order_p\left( \sqrt{\dfrac{\log n}{m}} \right)
\end{align*}
and hence
\begin{equation*}
\|\hat{X}_{S^c} - P_{S^c,S}U_{P_S}|D_{P_S}|^{-1/2}I_{p,q}W^\top\|_{2\to\infty}  \leq \|R_+\|_{2\to\infty} + \|R_-\|_{2\to\infty} = \order_p\left( \sqrt{\dfrac{\log n}{m}} \right)
\end{equation*}
Finally, $X_S(X_S^\top X_S)^{-1} = U_{P_S}|D_{P_S}|^{-1/2}$, which readily implies,
\begin{align*}
    \|\hat{X}_{S^c}W - X_{S^c}\|_{2\to\infty} & = \| \hat{X}_{S^c} - P_{S^c,S}X_S(X_S^\top X_S)^{-1}I_{p,q}W^\top \|_{2\to\infty} \\
    & = \|\hat{X}_{S^c} - P_{S^c,S}U_{P_S}|\Tilde{D}_{P_S}|^{-1/2}I_{p,q}W^\top\|_{2\to\infty} = \order_p\left( \sqrt{\dfrac{\log n}{m}} \right).
\end{align*}
In summary,  $\hat{X}_{PS} = [\hat{X}_S^{\top} \mid \hat{X}_{S^c}^{\top}]^{\top}$, we have
\begin{align*}
    \|\hat{X}_{PS}W-X_{PS}\|_{2\to\infty} & \leq \max\left\{ \|\hat{X}_{S}W-X_{S}\|_{2\to\infty} , \|\hat{X}_{S^c}W-X_{S^c}\|_{2\to\infty} \right\} = \order_p\left( \sqrt{\dfrac{\log n}{m}} \right)
\end{align*}
as claimed. 

\subsection{Proof of Theorem 2}
    For $W$ defined as in Eq. \eqref{W}, we can write
    \begin{align*}
        \hat{X}_S - U_{P_S}|D_{P_S}|^{1/2}W^\top & =\left[
            \hat{X}_{S+} - U_{P_S+}|D_{P_S+}|^{1/2}W_+^\top \,\,\, \Big | \,\,\, \hat{X}_{S-} - U_{P_S-}|D_{P_S-}|^{1/2}W_-^\top
        \right], \\
        \hat{X}_{S\pm} - U_{P_S\pm}|D_{P_S\pm}|^{1/2}W_\pm^\top 
        & = U_{A_S\pm}|D_{A_S\pm}|^{1/2} - U_{P_S\pm}W_\pm^\top|D_{A_S\pm}|^{1/2} + U_{P_S\pm}\tilde{Q}_{\pm},
    \end{align*}
    where $\tilde{Q}_{\pm} = W_\pm^\top|D_{A_S\pm}|^{1/2} - |D_{P_S\pm}|^{1/2}W_\pm^\top$. We know that,
    $$\||D_{P_S\pm}|^{1/2}\| = \order_p\left(\sqrt{\sigma_1(P_S)}\right) = \order_p(\sqrt{m\rho_n})$$
    Hence, by Eq.~\eqref{eq:lei_rinaldo} and Weyl's inequality (see Result~\ref{weyl}), we also have
    $$\||D_{A_S\pm}|^{1/2}\| \leq (\sigma_1(P_S) + \|A_S - P_S\|)^{1/2} = \order_p\left(\sqrt{m\rho_n}\right).$$
    Therefore, by Lemma B.3 in \cite{xie2024entrywise}, we have
    $$\|\tilde{Q}_{\pm}\|_F = \|W_\pm^\top|D_{A_S\pm}|^{1/2} - |D_{P_S\pm}|^{1/2}W_\pm^\top\|_F = \order_p\left( \dfrac{\sqrt{t}}{\sqrt{m\rho_n}} \right)$$
whenever $t \rightarrow \infty$. Choosing $t = \log m$ we get,
\begin{equation}
    \|\tilde{Q}_{\pm}\|_F = \order_p\left( \sqrt{\dfrac{\log m}{m\rho_n}} \right).
\end{equation}
Next, using Davis-Kahan's theorem and Lemma 6.7 in \cite{cape2019two} we have
\begin{equation}
    \| U_{A_S\pm} - U_{P_S\pm}W_\pm^\top\|_2 = \order_p\left( \dfrac{\|A_S - P_S\|_2}{m\rho_n\lambda_d(\Delta_m)} \right) = \order_p\left(\dfrac{1}{\sqrt{m\rho_n}}\right).
\end{equation}
See the derivations in Eq.~\eqref{eq:R2_bdd} for further details. 
We therefore have
    \begin{align*}
        \|\hat{X}_{S\pm} - U_{P_S\pm}|D_{P_S\pm}|^{1/2}W_\pm^\top \|_F & \leq \| U_{A_S\pm} - U_{P_S\pm}W_\pm^\top\| \times \||D_{A_S\pm}|^{1/2}\|_F + \|U_{P_S\pm}\| \times \|\tilde{Q}_{\pm}\|_F \\
        & = \order_p\left( \dfrac{1}{\sqrt{m\rho_n}} \times \sqrt{m\rho_n} + \sqrt{\dfrac{\log m}{m\rho_n}}  \right)  = \order_p(1). 
    \end{align*}
In summary, we have
\begin{align*}
\|\hat{X}_SW - X_S\|_F &= \|\hat{X}_S - U_{P_S}|D_{P_S}|^{1/2}W^\top\|_F
    \\ &= \sqrt{\|\hat{X}_{S+} - U_{P_S+}|D_{P_S+}|^{1/2}W_+^\top\|_F^2+\|\hat{X}_{S-} - U_{P_S-}|D_{P_S-}|^{1/2}W_-^\top\|_F^2}  \\
    & = \order_p\left( 1 \right).
\end{align*}
Next recall the decomposition for $\hat{X}_{S^c} - P_{S^c,S}U_{P_S}|D_{P_S}|^{-1/2}I_{p,q}W^\top$
in the proof of Theorem~1,
\begin{gather*} \hat{X}_{S^c} - P_{S^c,S}U_{P_S}|D_{P_S}|^{-1/2}I_{p,q}W^\top = \left[R_+ \mid - R_{-} \right], \\
R_\pm = E_{S^c,S} U_{A_S\pm} |D_{A_S\pm}|^{-1/2} 
    + P_{S^c,S} (U_{A_S\pm} - U_{P_S\pm}W_\pm^\top) |D_{A_S\pm}|^{-1/2} +  P_{S^c,S} U_{P_S\pm} Q_{2\pm}
\end{gather*}
where $Q_{2\pm} = W_{\pm}^\top |D_{A_S\pm}|^{-1/2} - |D_{P_S\pm}|^{-1/2} W_\pm^\top$. We thus have
\begin{align*}
    \|R_\pm\|_F & \leq (\|R_{1\pm}\|   
    + \|R_{2\pm}\|) \times \||D_{A_S\pm}|^{-1/2}\|_F +  \|P_{S^c,S} U_{P_S\pm}\| \times \|Q_{2\pm}\|_F
\end{align*}
where $R_{1\pm} = (A_{S^c,S} - P_{S^c,S})U_{A_S\pm}$ and $R_{2\pm} = P_{S^c,S} (U_{A_S\pm} - U_{P_S\pm}W_\pm^\top)$. Now 
$\|P_{S^c,S}\| \leq c_2 \rho_n\sqrt{mn}$. Next, by Theorem~5.2 in \cite{lei2015consistency},
$$\|R_{1\pm}\|_2 \leq \|A_{S^c,S} - P_{S^c,S}\|_2 \leq \|A-P\| = \order_p(\sqrt{n\rho_n}).$$
Applying Davis-Kahan's theorem \citep{davis70,samworth} and Lemma 6.7 in \cite{cape2019two} we have
$$\|R_{2\pm}\| \leq \|P_{S^c,S}\| \times \|U_{A_S\pm} - U_{P_S\pm}W_\pm^\top\|  = \order_p(\sqrt{n\rho_n}).$$
Combining the above bounds, we obtain
\begin{align*}
    \|R_{1\pm}\| \times \| |D_{A_S\pm}|^{-1/2}\|_F & = \order_p\left(\sqrt{n/m}\right),\\
    \|R_{2\pm}\| \times  \||D_{A_S\pm}|^{-1/2}\|_F & = \order_p\left(\sqrt{n/m}\right), \\
 \|P_{S^c,S} U_{P_S\pm}\| \times \|Q_{2\pm}\|_F & = \order_p\left(\sqrt{\dfrac{n\log m}{m^2\rho_n}}\right).
\end{align*}
and hence $\|R_\pm\|_F = \order_p(\sqrt{n/m})$. 
In summary, we have
\begin{equation*}
\begin{split}
    \|\hat{X}_{S^c}W - X_{S^c}\|_F &=  \| \hat{X}_{S^c} - P_{S^c,S}U_{P_S}|D_{P_S}|^{-1/2}I_{p,q}W^\top \|_F
    \\ &= \sqrt{\|R_+\|_F^2 + \|R_-\|_F^2} = \order_p\left(\sqrt{\dfrac{n}{m}}\right)
\end{split}
\end{equation*}
and hence
\begin{align*}
    \|\hat{X}_{PS}W-X_{PS}\|_{F} & \leq \sqrt{ \|\hat{X}_{S}W-X_{S}\|_F^2 + \|\hat{X}_{S^c}W-X_{S^c}\|_F^2} = \order_p\left(\sqrt{\dfrac{n}{m}}\right)
\end{align*}
as desired.

\subsection{Proof of Corollary 1}
For simplicity of notations we will henceforth assume, without loss of generality, that $\mathfrak{S} = I$.    Let $Y = \hat{X}_{PS} W - X_{PS}$. Then
    \begin{align*}
        \|\hat{P}_{PS} - P \|_{2\to\infty} & = \|(\hat{X}_{PS}W)I_{p,q}(\hat{X}_{PS}W)^\top - X_{PS}I_{p,q}X_{PS}^\top\|_{2\to\infty}\\
        & = \|(\hat{X}_{PS}W)I_{p,q}Y^{\top} + Y I_{p,q}X_{PS}^\top\|_{2\to\infty} \\
        & \leq \|\hat{X}_{PS}W\|_{2\to\infty} \times  \|Y\| + \|Y\|_{2\to\infty} \times \|X_{PS}\|  \\
        &  \leq \left( \|X_{PS}\|_{2\to\infty} + \|Y\|_{2 \to \infty} \right)\times \|Y\| + 
        \|Y\|_{2\to\infty} \times \|X_{PS}\|
        \\ & \leq \left( \|X_{PS}\|_{2\to\infty} + \|Y\|_{2 \to \infty} \right)\times \|Y\| + 
        \|Y\|_{2\to\infty} \times n^{1/2} \|X_{PS}\|_{2 \to \infty}
    \end{align*}
    Now by Theorem~1 and Theorem~2, we have
    \[\|Y\|_{2 \to \infty} = \order_{p}\left(\sqrt{\frac{\log n}{m}}\right), \qquad \|Y\| = \order_{p}\left(\sqrt{n/m}\right). \]
    Combining the above bounds we obtain
    \begin{equation*}
        \begin{split}
         \|\hat{P}_{PS} - P \|_{2\to\infty} &= \order_p\left(\left(\|X_{PS}\|_{2 \to \infty} + \sqrt{\dfrac{\log n}{m}}\right) \times \sqrt{\dfrac{n}{m}} + \sqrt{\dfrac{\log n}{m}} \times n^{1/2} \|X_{PS}\|_{2 \to \infty} \right) \\ &= \order_{p}\left(\|X_{PS}\|_{2 \to \infty} \sqrt{\frac{n \log n}{m}}\right).  
        \end{split}
    \end{equation*}
    Similarly we also have
    \begin{equation*}
    \begin{split}
        \|\hat{P}_{PS} - P \|_F
        & = \|(\hat{X}_{PS}W)I_{p,q}Y^{\top} - Y I_{p,q}X_{PS}^\top\|_F \\
        &  \leq \left( \|X_{PS}\|_F + \|Y\|_F \right)\|Y\| + \|Y\| \times \|X_{PS}\|_F \\
    & = \order_{p}\left(\frac{\|X_{PS}\|_{F} \sqrt{n}}{\sqrt{m}}\right). 
    \end{split}
    \end{equation*}

\subsection{Proof of Theorem 3}
    From Corollary 1, for any $c > 0$ there exists a constant $K_4 > 0$ and a $m_0$, both depending on $c$, such that for $i \in \{1,2\}$,
    \begin{equation}
        \label{eq:cor4_bdd}
    \|\hat{P}_{PS}^{(i)} - P^{(i)} \|_F \leq K_4 \|X_{PS}^{(i)}\|_{F} \sqrt{\dfrac{n}{m}}
    \end{equation}
    with probability at least $1 - m^{-c}$, provided that $m \geq m_0$.  Here $X_{PS}^{(i)} = \left[{X_S^{(i)}}^\top\mid {X_{S^c}^{(i)}}^\top\right]^\top$ is constructed based on the spectral embedding of $P_S^{(i)}$ followed by the predictive embedding $P_{S_c,S}^{(i)} X_S^{(i)}({X_S^{(i)}}^\top X_S^{(i)})^{-1} I_{p,q}$. Now suppose Eq.~\eqref{eq:cor4_bdd} holds. Then under $H_0 \colon P^{(1)} = P^{(2)}$ we have 
    \begin{align*}
 \frac{\sqrt{m}\|\hat{P}_{PS}^{(1)} - \hat{P}_{PS}^{(2)}\|_F}{\sqrt{n} K_4(\|X_{PS}^{(1)}\|_{F} + \|X_{PS}^{(2)}\|_{F})} & = \frac{\sqrt{m}\|\hat{P}_{PS}^{(1)} - P^{(1)} + P^{(2)} -\hat{P}_{PS}^{(2)}\|_F}{\sqrt{n} K_4(\|X_{PS}^{(1)}\|_{F} + \|X_{PS}^{(2)}\|_{F})} \\
        & \leq \frac{\sqrt{m}(\|\hat{P}_{PS}^{(1)} - P^{(1)}\|_F + \|\hat{P}_{PS}^{(2)} - P^{(2)}\|_F)}{\sqrt{n}(\|X_{PS}^{(1)}\|_{F} + \|X_{PS}^{(2)}\|_{F})} \leq 1
    \end{align*}
    Next, for $i=1,2$,
    \begin{align*}
        \|\hat{X}_{PS}^{(i)}\|_F \geq \|X_{PS}^{(i)}\|_F - \|\hat{X}_{PS}^{(i)} - X_{PS}^{(i)}\|_F = \|X_{PS}^{(i)}\|_{F} - \order_{p}\left(\sqrt{n/m}\right) \geq \frac{1}{2}\|X_{PS}^{(i)}\|_{F}
    \end{align*}
    Letting $R_F = \|\hat{X}_{PS}^{(1)}\|_F + \|\hat{X}_{PS}^{(2)}\|_F$ we have, for any fixed but arbitrary $\alpha > 0$, that
    $$T_F = \|\hat{P}_{PS}^{(1)} - \hat{P}_{PS}^{(2)}\|_F \leq 2K_4 R_F\sqrt{\dfrac{n}{m}}$$
    with probability at least $1-\alpha$.
    Next suppose $H_1$ is true and $\|P^{(1)}-P^{(2)}\|_F = \Omega(n \sqrt{\rho_n/m})$. Note that,
    $$\|\hat{P}_{PS}^{(1)} - \hat{P}_{PS}^{(2)}\|_F \geq \|P^{(1)}-P^{(2)}\|_F- \|\hat{P}_{PS}^{(1)} - P^{(1)}\|_F - \|\hat{P}_{PS}^{(2)} - P^{(2)}\|_F$$
    Letting $M = 2 K_4 R_F\sqrt{n/m}$ we have
    \begin{align*}
    \prob(T_F \leq M) & =  \prob\left(\|\hat{P}_{PS}^{(1)} - \hat{P}_{PS}^{(2)}\|_F \leq M\right)\\
    & \leq \prob\left(\|P^{(1)}-P^{(2)}\|_F- \|\hat{P}_{PS}^{(1)} - P^{(1)}\|_F - \|\hat{P}_{PS}^{(2)} - P^{(2)}\|_F \leq M\right) \\
        & = \prob\left( \|\hat{P}_{PS}^{(2)} - P^{(2)}\|_F + \|\hat{P}_{PS}^{(1)} - P^{(1)}\|_F \geq \|P^{(1)}-P^{(2)}\|_F - M \right).
    \end{align*}
    Now for any fixed but arbitrary $\beta > 0$, we have for $i=1,2$ and sufficiently large $m$ that
    \begin{align*}
        \prob\left(\|\hat{P}_{PS}^{(i)} - P^{(i)} \|_F \geq K_4 \|X_{PS}^{(i)}\|_{F} \sqrt{n/m} \right)& \leq \tfrac{\beta}{2} .
    \end{align*}
    We therefore have
    \[ \prob\left(\|\hat{P}_{PS}^{(1)} - P^{(1)} \|_F + \|\hat{P}_{PS}^{(2)} - P^{(2)} \|_F \geq M\right) 
    \leq \beta. \]       
    As $\|P^{(1)} - P^{(2)}\|_{F} = \Omega( n \sqrt{\rho_n/m})$ under $H_1$, there exists a sufficiently large constant $c$ such that 
    $\|P^{(1)}-P^{(2)}\|_F \geq c n \sqrt{\rho_n/m} > 2M$. Hence, $\prob(T_n \leq M) \leq \beta$, i.e., our test statistic $T_n$ lies within the rejection region with probability at least $1 - \beta$ as required.

\subsection{Proof of Theorem 4}

Note that, from Corollary 1, for any $c > 0$ there exists a constant $K_4 > 0$ and a $m_0$, both depending on $c$, such that for $i \in \{1,2\}$,
    \begin{equation}
        \label{eq:cor4_two}
    \|\hat{P}_{PS}^{(i)} - P^{(i)} \|_{2\to\infty} \leq K_4 \|X_{PS}^{(i)}\|_{2\to\infty} \sqrt{\dfrac{n\log n}{m}}
    \end{equation}
    with probability at least $1 - m^{-c}$, provided that $m \geq m_0$. Then, the rest of the proof can be completed following the exact same arguments used in the proof of Theorem~3.

\subsection{Proof of Theorem 5}
    We have,
    \begin{align*}
        \|\hat{P}_{S}^{(i)} - P_S^{(i)} \|_{2\to\infty} & = \|(\hat{X}_S^{(i)}W)I_{p,q}(\hat{X}_S^{(i)}W)^\top - X_SI_{p,q}X_S^\top\|_{2\to\infty}\\
        & = \|(\hat{X}_S^{(i)}W)I_{p,q}(\hat{X}_S^{(i)}W)^\top - (\hat{X}_S^{(i)}W)I_{p,q}X_S^\top\|_{2\to\infty} \\
        & \hspace{5cm}+ \|(\hat{X}_S^{(i)}W)I_{p,q}X_S^\top - X_SI_{p,q}X_S^\top\|_{2\to\infty} \\
        & \leq \|\hat{X}_S^{(i)}W\|_{2\to\infty} \|\hat{X}_S^{(i)}W - X_S\|_2+ \|\hat{X}_S^{(i)}W - X_S\|_{2\to\infty} \|X_S\|_2  \\
        &  \leq \left( \|X_S\|_{2\to\infty} + \|\hat{X}_S^{(i)}W - X_S\|_{2\to\infty} \right)\|\hat{X}_S^{(i)}W - X_S\|_F\\
        & \hspace{5cm}+ \|\hat{X}_S^{(i)}W - X_S\|_{2\to\infty}\sqrt{m}\|X_S\|_{2\to\infty} \\
        & \leq \order_p\left( \|X_S\|_{2\to\infty}\sqrt{\log m}  \right)
    \end{align*}
    i.e., for any $c > 0$ there exists a constant $C_1 > 0$ and a $m_0$, both depending on $c$, such that for $i \in \{1,2\}$,
    \begin{equation}
    \label{eq:puresub}
    \|\hat{P}_{S}^{(i)} - P_S^{(i)} \|_{2\to\infty} \leq C_1 \|X_{S}^{(i)}\|_{2\to\infty} \sqrt{\log m}
    \end{equation}
    with probability at least $1 - m^{-c}$, provided that $m \geq m_0$.  Now suppose Eq.~\eqref{eq:puresub} holds. Then under $H_0 \colon P^{(1)} = P^{(2)}$ we have 
    \begin{align*}
 \frac{\|\hat{P}_{S}^{(1)} - \hat{P}_{S}^{(2)}\|_{2\to\infty}}{\sqrt{\log m} C_1(\|X_{S}^{(1)}\|_{2\to\infty} + \|X_{S}^{(2)}\|_{2\to\infty})} & = \frac{\|\hat{P}_{S}^{(1)} - P_S^{(1)} + P_S^{(2)} -\hat{P}_{S}^{(2)}\|_{2\to\infty}}{\sqrt{\log m} C_1(\|X_{S}^{(1)}\|_{2\to\infty} + \|X_{S}^{(2)}\|_{2\to\infty})} \\
        & \leq \frac{(\|\hat{P}_{S}^{(1)} - P_S^{(1)}\|_{2\to\infty} + \|\hat{P}_{S}^{(2)} - P_S^{(2)}\|_{2\to\infty})}{\sqrt{\log m}(\|X_{S}^{(1)}\|_{2\to\infty} + \|X_{S}^{(2)}\|_{2\to\infty})} \leq 1
    \end{align*}
    Next, for $i=1,2$,
    \begin{align*}
        \|\hat{X}_{S}^{(i)}\|_{2\to\infty} &\geq \|X_{S}^{(i)}\|_{2\to\infty} - \|\hat{X}_{S}^{(i)} - X_{S}^{(i)}\|_{2\to\infty} \\ &= \|X_{S}^{(i)}\|_{2\to\infty} - \order_{p}\left(\sqrt{\dfrac{\log m}{m}}\right) \geq \frac{\|X_{S}^{(i)}\|_{2\to\infty}}{2}
    \end{align*}
Letting $R^S_{2\to\infty} = \|\hat{X}_{S}^{(1)}\|_{2\to\infty}  + \|\hat{X}_{S}^{(2)}\|_{2\to\infty} $ we have, for any fixed but arbitrary $\alpha > 0$, that
    $$T_{2\to\infty}^S = \|\hat{P}_{S}^{(1)} - \hat{P}_{S}^{(2)}\|_{2\to\infty} \leq 2C_1 R^S_{2\to\infty}\sqrt{\log m}$$
    with probability at least $1-\alpha$.

Now suppose $H_1$ is true i.e., there are at least $k \geq c_0n\log n/m$ rows in $P^{(1)}$ and $P^{(2)}$  with $\|P_i^{(1)}-P_i^{(2)}\|_2\geq c\sqrt{n\rho_n \log m/m}$ for some constants $c_0,c>0$. So for any randomly chosen subset $S$ of $m$ vertices we have,
\begin{align*}
    \prob\left(\max_{i\in S}\|P_i^{(1)}-P_i^{(2)}\|_2 < \|P^{(1)} - P^{(2)}\|_{2\to\infty}\right) & = \dfrac{\binom{n-k}{m}}{\binom{n}{m}} \\
    & = \prod_{i=1}^m \dfrac{n-k+1-i}{n-i+1} \\
    & = \exp\left\{ \sum_{i=1}^m \log \left( 1 - \dfrac{k}{n-i+1} \right)  \right\} \\
    & \leq \exp\left\{ -\sum_{i=1}^m  \dfrac{k}{n-i+1}  \right\} \\
    & \leq \exp\left\{ - \dfrac{mk}{n}  \right\} \\
    & \leq n^{-c_0} \\
    \implies \prob\left(\max_{i\in S}\|P_i^{(1)}-P_i^{(2)}\|_2 \geq \|P^{(1)} - P^{(2)}\|_{2\to\infty}\right) & \geq 1 - n^{-c_0}
\end{align*}

    Now $\|P_S^{(1)} -P_S^{(2)}\|_{2\to\infty} = \max_{i \in S} \|P_{Si}^{(1)} -P_{Si}^{(2)}\|_2$ where $P_{Si}$ denote the $i$th row of $P_S$ and we have for any subset $S$ of $m$ vertices,
    \begin{align*}
    \|P_{Si}^{(1)} -P_{Si}^{(2)}\|_2^2 & = \sum_{j=1}^n (P^{(1)}_{ij}-P^{(2)}_{ij})^2 Y_j \quad \text{where } Y_j = \mathbb{I}\{j\in S\} \thicksim \text{Bernoulli}(m/n)\quad \forall j=1(1)n 
\end{align*}
Now let $Z_j = (P^{(1)}_{ij}-P^{(2)}_{ij})^2 Y_j$ and we can write,
\begin{align*}
   |Z_j| & \leq (P^{(1)}_{ij}-P^{(2)}_{ij})^2 \leq c^2\rho_n^2 = M \quad \text{(say)} \\
   \E[Z_j] & = \dfrac{m}{n}\left(P^{(1)}_{ij}-P^{(2)}_{ij}\right)^2 \\
   \text{and  } \var[Z_j] & = (P^{(1)}_{ij}-P^{(2)}_{ij})^4 \var[Y_j] =\dfrac{m}{n}\left(1-\dfrac{m}{n}\right) (P^{(1)}_{ij}-P^{(2)}_{ij})^4 \leq \dfrac{mc^4\rho_n^4}{n}  = \sigma^2 \quad \text{(say)}
\end{align*}
So using Bernstein's inequality, we get,
\begin{align*}
    \prob\left( \sum_{j=1}^n Z_j - \sum_{j=1}^n \E[Z_j] < - t \right) & \leq \exp\left\{ - \dfrac{t^2}{2n\sigma^2 + 2t M/3 } \right\} \\
    \implies \prob\left( \|P_{Si}^{(1)} -P_{Si}^{(2)}\|_2^2 - \dfrac{m}{n}\|P_{i}^{(1)} -P_{i}^{(2)}\|_2^2 < - t \right)& \leq   \exp\left\{ - \dfrac{3t^2}{6mc^4\rho_n^4 + 2c^2\rho_n^2t}\right\} \\
\end{align*}
Choosing $t = k\rho_n^2(\sqrt{c_0m\log n} + c_0\log n)$, we get,
\begin{align*}
    \prob\left( \|P_{Si}^{(1)} -P_{Si}^{(2)}\|_2^2 \leq \dfrac{m}{n}\|P_{i}^{(1)} -P_{i}^{(2)}\|_2^2 (1-\epsilon) \right) & \leq   n^{-c_0}\\
    \implies \prob\left( \|P_{Si}^{(1)} -P_{Si}^{(2)}\|_2 \leq \sqrt{\dfrac{m}{n}}\|P_{i}^{(1)} -P_{i}^{(2)}\|_2 \sqrt{1-\epsilon} \right) & \leq   n^{-c_0} \\
    \implies \prob\left( \|P_{Si}^{(1)} -P_{Si}^{(2)}\|_2 \geq \sqrt{\dfrac{m}{n}}\|P_{i}^{(1)} -P_{i}^{(2)}\|_2 \sqrt{1-\epsilon} \right)& \geq 1-  n^{-c_0} \\
    \implies \prob\left( \max_{i\in S}\|P_{Si}^{(1)} -P_{Si}^{(2)}\|_2 \geq \sqrt{\dfrac{m}{n}}\max_{i\in S}\|P_{i}^{(1)} -P_{i}^{(2)}\|_2 \sqrt{1-\epsilon} \right)& \geq 1-  n^{-c_0}
\end{align*}
Combining we get,
\begin{align*}
    \prob\left( \|P_{S}^{(1)} -P_{S}^{(2)}\|_{2\to\infty} \geq \sqrt{\dfrac{m}{n}}\|P_{i}^{(1)} -P_{i}^{(2)}\|_{2\to\infty} \sqrt{1-\epsilon} \right)& \geq 1-  2n^{-c_0}
\end{align*}
 Note that,
    $$\|\hat{P}_{S}^{(1)} - \hat{P}_{S}^{(2)}\|_{2\to\infty} \geq \|P_S^{(1)}-P_S^{(2)}\|_{2\to\infty}- \|\hat{P}_{S}^{(1)} - P_S^{(1)}\|_{2\to\infty} - \|\hat{P}_{S}^{(2)} - P_S^{(2)}\|_{2\to\infty}$$
 Letting $M = 2 C_1 R_S\sqrt{\log m}$ we have
    \begin{align*}
    \prob(T_{2\to\infty}^S \leq M) & =  \prob\left(\|\hat{P}_{S}^{(1)} - \hat{P}_{S}^{(2)}\|_{2\to\infty} \leq M\right)\\
    & \leq \prob\left(\|P_S^{(1)}-P_S^{(2)}\|_{2\to\infty}- \|\hat{P}_{S}^{(1)} - P_S^{(1)}\|_{2\to\infty} - \|\hat{P}_{S}^{(2)} - P_S^{(2)}\|_{2\to\infty} \leq M\right) \\
        & = \prob\left( \|\hat{P}_{S}^{(2)} - P_S^{(2)}\|_{2\to\infty} + \|\hat{P}_{S}^{(1)} - P_S^{(1)}\|_{2\to\infty} \geq \|P_S^{(1)}-P_S^{(2)}\|_{2\to\infty} - M \right).
    \end{align*}
    Now for any fixed but arbitrary $\beta > 0$, we have for $i=1,2$ and sufficiently large $m$ that
    \begin{align*}
        \prob\left(\|\hat{P}_{S}^{(i)} - P_S^{(i)} \|_{2\to\infty} \geq c_1 \|X_{S}^{(i)}\|_{2\to\infty} \sqrt{\log m} \right)& \leq \tfrac{\beta}{2} .
    \end{align*}
    We therefore have
    $$\prob\left(\|\hat{P}_{S}^{(1)} - P_S^{(1)} \|_{2\to\infty} + \|\hat{P}_{S}^{(2)} - P_S^{(2)} \|_{2\to\infty} \geq M\right) 
    \leq \beta.$$       
    By the assumption on $H_1$ in the theorem, $\|P^{(1)}-P^{(2)}\|_{2\to\infty}>c\sqrt{n\rho_n\log m/m}$ for some constant $c$ i.e., $\|P_{S}^{(1)}-P_{S}^{(2)}\|_{2\to\infty}>2M$. Hence, $\prob(T_n \leq M) \geq  \beta$ i.e., our test statistic $T_{2\to\infty}^S$ lies within the rejection region with probability at least $1 - \beta$, as required.

\subsection{Proof of Theorem~6}

We have,
\begin{align*}
   \left \|\hat{P}_S^{(i)} - P_S^{(i)}\right\|_F 
    &= \left\| \hat{X}_S^{(i)}WI_{p,q} \left(\hat{X}_S^{(i)}W\right)^\top - X_S^{(i)} I_{p,q}\left(X_S^{(i)}\right)^\top \right\|_F \\
    &= \left\| \hat{X}_S^{(i)}WI_{p,q} \left(\hat{X}_S^{(i)}W\right)^\top - \hat{X}_S^{(i)} W I_{p,q} \left(X_S^{(i)}\right)^\top + \hat{X}_S^{(i)} WI_{p,q}\left(X_S^{(i)}\right)^\top - X_S^{(i)} I_{p,q}\left(X_S^{(i)}\right)^\top \right\|_F \\
    &\leq \left\| \hat{X}_S^{(i)}W - X_S^{(i)} \right\|_F \left( \left\| \hat{X}_S^{(i)} \right\|_F + \left\| X_S^{(i)} \right\|_F \right) \\
    & \leq \left\| \hat{X}_S^{(i)}W - X_S^{(i)} \right\|_F \left( \left\| \hat{X}_S^{(i)}W - X_S^{(i)} \right\|_F + 2\left\| X_S^{(i)} \right\|_F \right) \\
    & \leq \order_p\left( \|X_S^{(i)}\|_F \right)
\end{align*}
i.e., for any $c > 0$ there exists a constant $C_1 > 0$ and a $m_0$, both depending on $c$, such that for $i \in \{1,2\}$,
\begin{equation}
\label{eq:puresubF}
\|\hat{P}_{S}^{(i)} - P_S^{(i)} \|_F \leq C_1 \|X_{S}^{(i)}\|_F
\end{equation}
with probability at least $1 - m^{-c}$, provided that $m \geq m_0$. Now suppose Eq.~\eqref{eq:puresubF} holds. Then under $H_0 \colon P^{(1)} = P^{(2)}$ we have
\begin{align*}
\frac{\|\hat{P}_{S}^{(1)} - \hat{P}_{S}^{(2)}\|_F}{C_1(\|X_{S}^{(1)}\|_F + \|X_{S}^{(2)}\|_F)} & = \frac{\|\hat{P}_{S}^{(1)} - P_S^{(1)} + P_S^{(2)} -\hat{P}_{S}^{(2)}\|_F}{C_1(\|X_{S}^{(1)}\|_F + \|X_{S}^{(2)}\|_F)} \\
    & \leq \frac{(\|\hat{P}_{S}^{(1)} - P_S^{(1)}\|_F + \|\hat{P}_{S}^{(2)} - P_S^{(2)}\|_F)}{(\|X_{S}^{(1)}\|_F + \|X_{S}^{(2)}\|_F)} \leq 1
\end{align*}
Next, for $i=1,2$,
    \begin{align*}
        \|\hat{X}_{S}^{(i)}\|_F &\geq \|X_{S}^{(i)}\|_F - \|\hat{X}_{S}^{(i)} - X_{S}^{(i)}\|_F \\ &= \|X_{S}^{(i)}\|_F - \order_{p}\left(1\right)\\
        & =  \|X_{S}^{(i)}\|_F - o_{p}\left(\|X_{S}^{(i)}\|_F\right) \geq \frac{\|X_{S}^{(i)}\|_F}{2}
    \end{align*}
Letting $R^S_F = \|\hat{X}_{S}^{(1)}\|_F  + \|\hat{X}_{S}^{(2)}\|_F$ we have, for any fixed but arbitrary $\alpha > 0$, that
$$T_{F}^S = \|\hat{P}_{S}^{(1)} - \hat{P}_{S}^{(2)}\|_F \leq 2C_1 R^S_F$$
with probability at least $1-\alpha$.

Now suppose $H_1$ is true i.e., $\|P^{(1)}-P^{(2)}\|_F>cn^2\sqrt{\rho_n/m^3}$ for some constant $c>0$. So we have,
\begin{align*}
    \|P_S^{(1)} -P_S^{(2)}\|_F^2 & = \sum_{i=1}^n \sum_{j=1}^n (P^{(1)}_{ij}-P^{(2)}_{ij})^2 Y_iY_j \quad \text{where } Y_i = \mathbb{I}\{i\in S\} \thicksim \text{Bernoulli}(m/n)\quad \forall i=1(1)n \\
    & = \boldsymbol{Y}^\top D \boldsymbol{Y} \quad \text{with  } \boldsymbol{Y} = (Y_1,Y_2,\cdots,Y_n),\quad  (D_{n\times n})_{ij} = (P^{(1)}_{ij}-P^{(2)}_{ij})^2
\end{align*}
Now let $Z_i = Y_i - \dfrac{m}{n}$. So we have,
\begin{align*}
    \E[\boldsymbol{Y}^\top D \boldsymbol{Y}] & = \|P^{(1)}-P^{(2)}\|_F^2 \dfrac{m^2}{n^2} \\
    \boldsymbol{Y}^\top D \boldsymbol{Y} - \E\left[\boldsymbol{Y}^\top D \boldsymbol{Y}\right] & = \left(\boldsymbol{Z}+\frac{m}{n}\boldsymbol{1}\right)^\top D\left(\boldsymbol{Z}+\frac{m}{n}\boldsymbol{1}\right) - \dfrac{m^2}{n^2}\boldsymbol{1}^\top D \boldsymbol{1} \\
    & = \boldsymbol{Z}^\top D \boldsymbol{Z} + \dfrac{2m}{n}\boldsymbol{Z}^\top D\boldsymbol{1}
\end{align*}
Now \cite{he2024sparse} mentioned in (2.3) of their paper that a special case of Theorem 3.1 in \cite{gine2000exponential}'s work yields,
\begin{align*}
     \prob\left( \left| Z^\top D Z \right| > t \right) & \leq 2\exp \left\{ -c\min \left\{ 
 \dfrac{n^2t^2}{m^2\sum_{i\ne j} D_{ij}}, \dfrac{nt}{m\|D\|}, \left( \dfrac{t}{B}\right)^{2/3}, \left( \dfrac{t}{\|D\|_{\max}} \right)^{1/2}  \right\} \right\}
\end{align*}
where $B = \max_j \sqrt{\dfrac{m}{n}\sum_i D_{ij}^2} = \Theta(\sqrt{m}\rho_n^2)$, $\|D\|_{\max} = \max_{i,j} D_{ij} = \Theta(\rho_n^2)$, and $\|D\|_F = \Theta(n\rho_n^2)$. Choosing $t = cm\rho_n^2\log n$, we get,
$$\prob\left( Z^\top D Z < - t \right) \leq 2n^{-c_0}$$
for some $c_0>0$. Again, we have,
$$\boldsymbol{Z}^\top D\boldsymbol{1} = \sum_{i=1}^n \sum_{j=1}^n (P^{(1)}_{ij}-P^{(2)}_{ij})^2 Z_i = \sum_{i=1}^n C_iZ_i \quad \text{where  } C_i =  \sum_{j=1}^n (P^{(1)}_{ij}-P^{(2)}_{ij})^2$$
Let $W_i = \dfrac{2m}{n}C_iZ_i$. Thus we can get,
\begin{align*}
    \E\left[ W_i \right] & = 0 \quad \text{and  } \E\left[ W_i^2 \right] = \dfrac{4m^2}{n^2}C_i^2\E[Z_i^2] = \dfrac{m}{n}\left(1-\dfrac{m}{n}\right) C_i^2 = \Theta\left(\dfrac{m^3\rho_n^4}{n}\right)
\end{align*}
So using Bernstein's inequality, we get,
\begin{align*}
    \prob\left( \sum_{i=1}^n W_i > t \right) & \leq \exp\left\{ -\dfrac{t^2/2}{\displaystyle\sum_{i=1}^n \E[W_i^2] + \dfrac{t}{3}}  \right\}
\end{align*}
Choosing $t = cm\rho_n^2\log n$, we get,
$$\prob\left( Z^\top D \boldsymbol{1} > - t \right) \geq 1 - n^{-c_0}$$
Combining we get,
\begin{align*}
    \prob\left( \|P_S^{(1)} -P_S^{(2)}\|_F^2 > \|P^{(1)}-P^{(2)}\|_F^2 \dfrac{m^2}{n^2}(1-\epsilon) \right) & \geq 1 - 2n^{-c_0}
\end{align*}
Note that,
$$\|\hat{P}_{S}^{(1)} - \hat{P}_{S}^{(2)}\|_F \geq \|P_S^{(1)}-P_S^{(2)}\|_F- \|\hat{P}_{S}^{(1)} - P_S^{(1)}\|_F - \|\hat{P}_{S}^{(2)} - P_S^{(2)}\|_F$$
Letting $M = 2C_1 R_S$ we have
\begin{align*}
\prob(T_{F}^S \leq M) & =  \prob\left(\|\hat{P}_{S}^{(1)} - \hat{P}_{S}^{(2)}\|_F \leq M\right)\\
& \leq \prob\left(\|P_S^{(1)}-P_S^{(2)}\|_F- \|\hat{P}_{S}^{(1)} - P_S^{(1)}\|_F - \|\hat{P}_{S}^{(2)} - P_S^{(2)}\|_F \leq M\right) \\
    & = \prob\left( \|\hat{P}_{S}^{(2)} - P_S^{(2)}\|_F + \|\hat{P}_{S}^{(1)} - P_S^{(1)}\|_F \geq \|P_S^{(1)}-P_S^{(2)}\|_F - M \right).
\end{align*}
Now for any fixed but arbitrary $\beta > 0$, we have for $i=1,2$ and sufficiently large $m$ that
\begin{align*}
    \prob\left(\|\hat{P}_{S}^{(i)} - P_S^{(i)} \|_F \geq c_1 \|X_{S}^{(i)}\|_F \right)& \leq \tfrac{\beta}{2} .
\end{align*}
We therefore have
$$\prob\left(\|\hat{P}_{S}^{(1)} - P_S^{(1)} \|_F + \|\hat{P}_{S}^{(2)} - P_S^{(2)} \|_F \geq M\right) 
\leq \beta.$$
By the assumption on $H_1$ in the theorem, $\|P^{(1)}-P^{(2)}\|_F>cn^2\sqrt{\rho_n/m^3}$ for some constant $c$ i.e., $\|P_{S}^{(1)}-P_{S}^{(2)}\|_F>2M$. Hence, $\prob(T_{F}^S \leq M) \leq  \beta$ i.e., our test statistic $T_{F}^S$ lies within the rejection region with probability at least $1 - \beta$, as required.

\end{document}